\documentclass[12pt]{article}

\usepackage[utf8]{inputenc}
\usepackage[margin=1in]{geometry} 
\usepackage{amsmath,amsthm,amssymb, graphicx, multicol, array, sgame, setspace, threeparttable, tabularx, rotating, booktabs, multirow}
\usepackage{xltabular} 
\usepackage{dcolumn} 
\usepackage[square,sort&compress, authoryear, nonamebreak, round]{natbib}
\usepackage[table]{xcolor}
\usepackage{titlesec}

\usepackage{listings}
\usepackage[breakable,skins,listings]{tcolorbox}
\definecolor{promptbg}{gray}{0.95}
\newtcblisting{promptbox}{
  breakable, enhanced, colback=promptbg, colframe=promptbg,
  arc=2.5pt, boxrule=0pt, left=2mm, right=2mm, top=1.5mm, bottom=1.5mm,
  listing only,
  listing options={basicstyle=\ttfamily\footnotesize, breaklines=true,
    columns=fullflexible, keepspaces=true}
}

\usepackage{makecell}
\usepackage{bbm}
\usepackage{float}
\usepackage{mathtools}
\setcitestyle{round}
\usepackage{tikz}
\usepackage{pgfplots}
\usetikzlibrary{plotmarks}
\usepackage[ruled, linesnumbered]{algorithm2e}

\usepackage{ulem}
\definecolor{sangre}{rgb}{0.6,0.18,0.19}
\definecolor{dullmagenta}{rgb}{0.4,0,0.4}
\definecolor{darkblue}{rgb}{0,0,0.6}

\ifx\pdfoutput\undefined
\usepackage[backref=none,hypertex,colorlinks,urlcolor = darkblue,citecolor = sangre,linkcolor=blue]{hyperref}
\else
\usepackage[backref=none,hypertexnames=false,colorlinks,urlcolor = darkblue,citecolor =  sangre,linkcolor=blue]{hyperref}
\fi

\usepackage[capitalize]{cleveref}

\usepackage{color-edits}
\addauthor{yl}{red}
\addauthor{xq}{brown}
\addauthor{hf}{blue}
\addauthor{yy}{teal}

\usepackage{algorithmicx}
\algnewcommand\COMMENT[1]{\hfill\(\triangleright\) #1}

\usepackage{subcaption}
\usepackage{color,soul}

\usetikzlibrary{decorations.pathreplacing} 

\usepackage{enumitem}
\newlist{steps}{enumerate}{1}
\setlist[steps, 1]{label = Step \arabic*:}

\providecommand{\keywords}[1]
{
  \small	
  \textbf{\textit{Keywords---}} #1
}

\hypersetup{
    colorlinks=true,
    linkcolor=blue,
    filecolor=magenta,      
    urlcolor=cyan,
    pdftitle={Overleaf Example},
    pdfpagemode=FullScreen,
    }

\title{Performance Manipulation: Labor Market Implications in AI-assisted Era
}

\author{
 Xiaoyun Qiu\textsuperscript{1}\textsuperscript{\ensuremath{\dagger}}
\and Yang Yu\textsuperscript{2}\textsuperscript{\ensuremath{\dagger}}
\and Haifeng Xu\textsuperscript{3}
}
\date{\today}

\newcommand{\type}{\theta}
\newcommand{\types}{\boldsymbol{\type}}
\newcommand{\typesminus}{\types_{-i}}
\newcommand{\budget}{e}
\newcommand{\budgetbenchmark}{\budget^{\dag}}
\newcommand{\budgetcontest}{\budget^*}
\newcommand{\mdeffort}{a}
\newcommand{\pteffort}{b}

\newcommand{\mdeffortbenchmark}{\mdeffort^{\dagger}}
\newcommand{\pteffortbenchmark}{\pteffort^{\dagger}}
\newcommand{\mdeffortcontest}{\mdeffort^{*}}
\newcommand{\pteffortcontest}{\pteffort^{*}}

\newcommand{\fitness}{\mu}
\newcommand{\fitnesses}{\boldsymbol{\fitness}}
\newcommand{\fitnesscontest}{\fitness^*}
\newcommand{\fitnessbenchmark}{\fitness^{\dagger}}

\newcommand{\lfitness}{\fitness'}
\newcommand{\ufitness}{\fitness}

\newcommand{\ltype}{\type'}
\newcommand{\utype}{\type}

\newcommand{\lgamma}{\gamma'}
\newcommand{\ugamma}{\gamma}

\newcommand{\eqmfitness}{\beta}
\newcommand{\eqmfitnessalt}{\fitness}
\newcommand{\eqmfitnessesalt}{\boldsymbol{\eqmfitnessalt}}

\newcommand{\eqmeffortalt}{\phi}
\newcommand{\eqmeffortalts}{\boldsymbol{\eqmeffortalt}}

\newcommand{\eqmfitnesses}{\boldsymbol{\eqmfitness}}

\newcommand{\signal}{s}
\newcommand{\signals}{\boldsymbol{\signal}}

\newcommand{\reward}{R}
\newcommand{\rewards}{\boldsymbol{\reward}}

\newcommand{\numplayer}{N}

\newcommand{\mdeffortfn}{\nu}
\newcommand{\pteffortfn}{\xi}
\newcommand{\etcostfn}{c}
\newcommand{\fscostfn}{C}
\newcommand{\contestgain}{g}
\newcommand{\payoffbm}{\pi^b}
\newcommand{\typebl}{\type_1^{\dagger}}
\newcommand{\typebu}{\type_2^{\dagger}}
\newcommand{\typecl}{ \type_1^*}
\newcommand{\typecu}{ \type_2^*}

\newcommand{\pmf}{p}
\newcommand{\cmf}{P}

\newcommand{\payoff}{\pi}

\newcommand{\disttype}{F}
\newcommand{\densitytype}{f}

\newcommand{\distsignal}{H}
\newcommand{\densitysignal}{h}

\newcommand{\given}{|}

\newcommand{\prob}[2][]{\text{\bf Pr}\ifthenelse{\not\equal{}{#1}}{_{#1}}{}\!\left[{\def\givenn{\middle|}#2}\right]}
\newcommand{\expect}[2][]{\text{\bf E}\ifthenelse{\not\equal{}{#1}}{_{#1}}{}\!\left[{\def\givenn{\middle|}#2}\right]}

\newcommand{\tparen}{\big}
\newcommand{\tprob}[2][]{\text{\bf Pr}\ifthenelse{\not\equal{}{#1}}{_{#1}}{}\tparen[{\def\given{\tparen|}#2}\tparen]}
\newcommand{\texpect}[2][]{\text{\bf E}\ifthenelse{\not\equal{}{#1}}{_{#1}}{}\tparen[{\def\given{\tparen|}#2}\tparen]}

\newcommand{\sprob}[2][]{\text{\bf Pr}\ifthenelse{\not\equal{}{#1}}{_{#1}}{}[#2]}
\newcommand{\sexpect}[2][]{\text{\bf E}\ifthenelse{\not\equal{}{#1}}{_{#1}}{}[#2]}

\DeclareMathOperator*{\argmax}{arg\,max}
\DeclareMathOperator*{\argmin}{arg\,min}

\newcommand{\indicate}[1]{{\bf 1}\left[#1\right]}

\theoremstyle{definition}
\newtheorem{definition}{Definition}

\newtheorem{example}{Example}
\newtheorem{assumption}{Assumption}

\theoremstyle{plain}

\newtheorem{corollary}{Corollary}
\newtheorem{lemma}{Lemma}
\newtheorem{proposition}{Proposition}

\theoremstyle{remark}

\begin{document}
\maketitle

\begingroup
\renewcommand{\thefootnote}{\ensuremath{\dagger}}
\footnotetext[4]{Equal contribution, alphabetically ordered by last name. Haifeng Xu is supported by AI2050 program at Schmidt Sciences (Grant G-24-66104), NSF through CCF-2303372 and ARO through W911NF-23-1-0030. We acknowledge Yingkai Li for helpful conversations at the early stage of this project. We also thank Ben Brooks, Lei Huang, Kyohei Okumura, Abhishek Sarkar, Chenhao Zhang, Chunru Zheng and participants at 2024 DC IO Day, 2025 APIOC and 2026 IIOC for helpful comments and suggestions.}
\endgroup
\footnotetext[1]{Department of Economics, Dartmouth College.
Email: \texttt{xiaoyun.qiu@dartmouth.edu}}
\footnotetext[2]{Sloan \& CSAIL, Massachusetts Institute of Technology.
Email: \texttt{yyu7@mit.edu}}
\footnotetext[3]{Department of Computer Science, University of Chicago.
Email: \texttt{haifengxu@uchicago.edu}}

\begin{abstract}

\emph{Performance manipulation} arises when agents exploit easily measurable,
routine tasks to inflate observable outcomes without contributing genuine
innovation or expert judgment. We formalize this phenomenon in a game-theoretic model in which agents allocate effort along
two margins. \emph{Creative effort} is non-routine cognitive labor whose return is
complementary to the agent's private expertise; it is the scarce input that
principals seek. \emph{Mechanistic effort} is the execution of well-defined,
rule-based tasks that raise performance independently of expertise, a
commoditized input that AI heavily augments. We establish the existence of a
symmetric, monotone pure-strategy equilibrium and show that performance-based
screening remains viable so long as evaluations retain a sufficient creative
component, but collapses into an uninformative pooling equilibrium once AI
capability grows large enough to crowd out creative effort. Comparing contest
allocations against a single-agent baseline isolates performance manipulation as
the competition-induced over-investment in mechanistic effort, which we show is
undertaken systematically by low-type agents but not high-type ones. We further
prove that more sharply skewed reward structures mitigate this friction by
eliciting greater creative effort across the participant pool. Finally, using a
novel, language-model-based methodology to measure both effort types from nearly
1{,}500 Kaggle competition scripts, we provide robust empirical support for the
model's predictions.
  \end{abstract}

\keywords{contests, performance manipulation, gaming, artificial intelligence, labor market}

\section{Introduction}\label{sec:intro}

The pervasive adoption of generative artificial intelligence (AI) has
precipitated a crisis of confidence in labor-market screening. For decades,
employers have relied on asynchronous evaluations---take-home coding
assessments, data-analysis assignments, legal-research memoranda---to gauge a
candidate's underlying ability. However, as AI models attain proficiency on such
standardized tasks, these instruments are losing their diagnostic
value: if a candidate can produce a competent deliverable in minutes with an AI
assistant, the observable quality of that output no longer reliably signals
private expertise. The phenomenon spans industries: software developers use
tools such as GitHub Copilot to generate standard code \citep{peng2023impact},
law firms deploy AI for e-discovery and boilerplate drafting \citep{remus2017can},
and researchers employ language models to conduct literature reviews and routine
regressions \citep{korinek2023language}. Moreover, AI
augments performance on standardized tasks far more than on the non-routine
judgment that underlies genuine expertise \citep{autor2025expertise}. This asymmetry raises a core
principal-agent concern, which we term \emph{performance manipulation}: agents
may exploit AI to inflate observable performance metrics without contributing to 
genuine, expertise-driven innovation. Motivated by these shifts, this paper asks:
how do AI tools change the efficacy and design of performance-based screening in
the labor market?

To answer this, we must first decompose the nature of work in the AI
era. 
Our taxonomy extends the influential framework of \citet{autor2003skill}, who, in examining how computer technology alters job skill demands, partitioned tasks according to whether their procedures could be codified into explicit rules: \emph{routine} tasks, which can be programmed as an unambiguous, ordered sequence of computational instructions that produce deterministic outcomes, and \emph{non-routine} tasks, which rely on human problem-solving and judgment.
The boundary, however, shifts with technological advancement. When
\citet{autor2003skill} wrote, routine effectively meant reducible to the
explicit, procedural logic of conventional software, and a large class of
rule-governed cognitive tasks, including drafting standardized text, summarizing
documents, translating specifications into boilerplate, couldn't be readily codified and therefore remained the domain of human labor. Large language
models, which acquire task competence inductively from data rather than through
hand-specified rules, have relocated this frontier, becoming capable substitutes
for many of today's routine cognitive tasks that previously demanded human
execution. Our distinction between mechanistic and creative effort 
inherits the substitution-versus-complementarity logic of \citet{autor2003skill}
while repartitioning cognitive labor along the boundary that contemporary AI is able to cross.

Building on this, we distinguish between two fundamentally different types of
effort in the era of AI. \emph{Mechanistic effort} denotes the execution of
well-defined, standardized tasks that can be accomplished by applying explicit
rules without recourse to human judgment.\footnote{A growing body of evidence
indicates that current AI systems act as powerful substitutes for mechanistic
effort, compressing the output distribution by raising its lower tail.
\citet{peng2023impact} document large speed-ups among developers using GitHub
Copilot on well-specified coding tasks, while \citet{genAI2025} find that
generative AI in customer support markedly raised the productivity of novice
agents by diffusing routine problem-solving procedures previously held tacitly
by experienced workers, but yielded comparatively modest gains for top
performers.} In software engineering, it includes writing boilerplate code; in
legal practice, summarizing existing case law; and in machine learning,
brute-force hyperparameter tuning. Mechanistic effort does not reply on the
agent's private expertise and is, by construction, highly substitutable by
AI capabilities.

In contrast, \emph{creative effort} denotes the open-ended cognitive labor of
problem formulation, strategic judgment, and novel synthesis for which no
predetermined rules exist.\footnote{At the frontier of task complexity, the
empirical record shows that AI continues to require human direction.
\citet{bubeck2023paper} observe that GPT-4, though fluent at syntactic and
retrieval-style reasoning, still depends on human input for architectural design
and problem decomposition; \citet{dellacqua2023navigating} document a ``jagged
technological frontier'' among management consultants, on which AI sharply
improves performance within its capability envelope but degrades it on tasks
requiring context-heavy, integrative reasoning. In short, current AI raises the
performance floor without lowering the creative ceiling.} Examples are the
design of unprecedented system architectures, the formulation of novel
litigation strategies, and the development of new algorithmic
frameworks (see Appendix~\ref{apx:effort_examples_four_applications} for detailed examples).  
While AI delivers substantial gains on the mechanistic margin,
output at the creative frontier remains strictly dependent on humans' domain
expertise. Later, we will also see that AI's substitutability for mechanistic
effort but not creative effort opens the
door to performance manipulation and motivates the screening problem we study.


We find striking empirical support for this separation of effort using data from machine learning competitions on Kaggle.
Unlike previous studies that rely on indirect and heuristic proxies for effort, such as working hours, we develop direct measures of mechanistic and creative effort grounded in contestants' actual problem-solving behavior, by using a large language model to analyze the content of contestants’ submitted code. We validate these measures through human verification and independent evaluations by alternative large language models.
\cref{fig:effort_by_type_intro} reveals a stark divergence: high-ability participants exert roughly twice as much creative effort as low-ability participants, but only marginally more mechanistic effort. This pattern is confirmed by our regression analysis (columns (2)(3) in \cref{tab:regressions}): regressing creative effort on participant type yields a positive and significant coefficient, indicating that high-type participants systematically engage in more creative problem-solving.  Conversely, regressing mechanistic effort on type yields a small, negative, and insignificant coefficient. High-ability and low-ability agents alike rely heavily on mechanistic techniques, confirming that such effort functions as a broadly accessible baseline rather than a differentiator of talent.

\begin{figure}[t]
    \caption{\bf Creative and Mechanistic Effort by Type}
	\centering
	\includegraphics[width=0.85\textwidth]{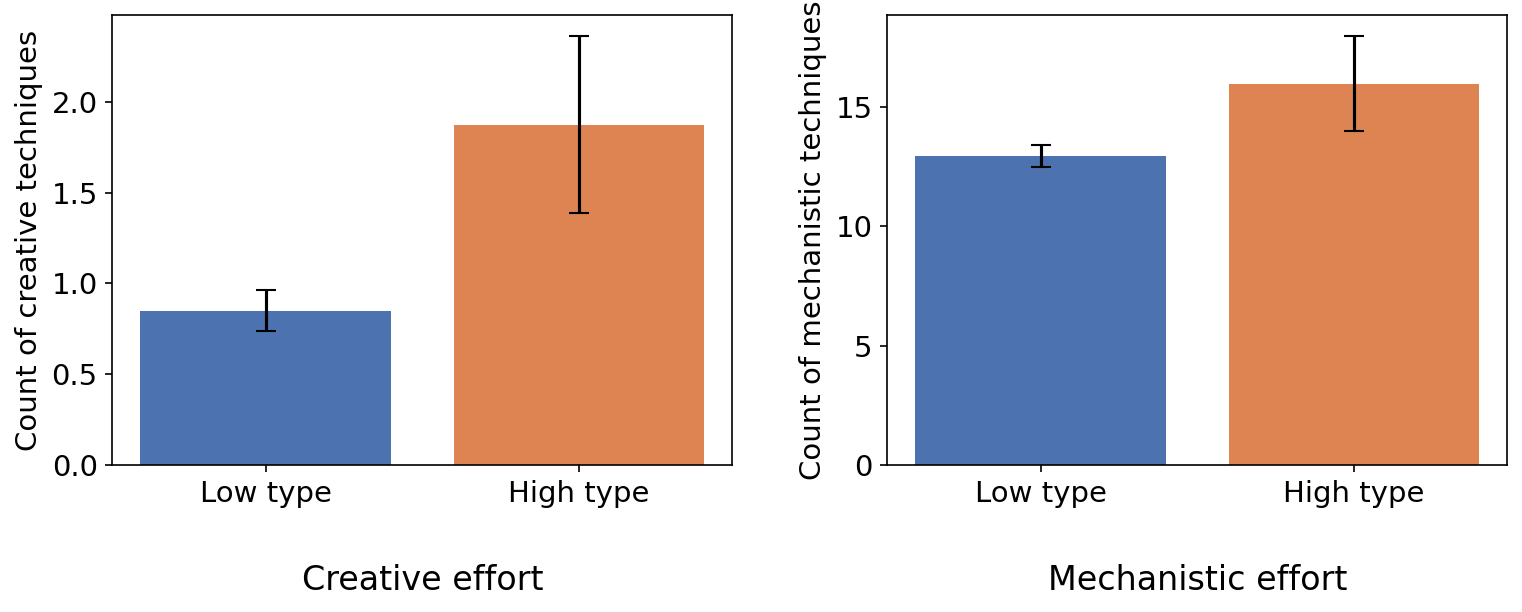}
	\label{fig:effort_by_type_intro}
    \caption*{\footnotesize \emph{Note:} The magnitude of creative (mechanistic) effort is measured by the number of creative (mechanistic) techniques identified from the submitted code scripts. Detailed empirical methodology is explained in \cref{sec:fact}. The error bars show 95\% confidence interval of the technique count.}
\end{figure}

To formally capture these empirical observations, we develop a game-theoretic contest model of performance-based screening. A principal hosts a contest where players compete on the overall \emph{performance} of their output. The contest set up can capture settings like a hiring evaluation, a grant race, or a promotional tournament. Each player possesses a private type representing their underlying capability and allocates effort between creation and mechanization. While both efforts improve observable performance, the principal desires creative effort, which drives true innovation and generalizability. Thus, the principal faces a multi-task moral hazard problem     à la \citet{holmstrom1991multitask}: the principal only observes aggregate performance, not the composition of effort. Our model formalizes the AI shock by assuming that mechanistic productivity is enhanced by an exogenous AI parameter, whereas the return to creative effort is strictly supermodular with the agent's private type.

Analyzing the equilibrium of this game allows us to characterize the conditions under which performance-based screening remains theoretically viable, and when it inevitably collapses. In the canonical signaling framework of \citet{spence1973}, an observable signal is informative if and only if it exhibits the strict single-crossing property, meaning it must be differentially less costly for high-type agents to produce. We establish that when AI capabilities become sufficiently advanced, the marginal productivity of mechanistic effort becomes so high that it completely crowds out human creative effort. Because all agents optimally rely entirely on this AI-enhanced mechanistic margin, the cost of performance becomes independent of inherent ability, destroying the single-crossing property and collapsing the market into a pooling equilibrium. As a result, performance-based screening becomes ineffective. Conversely, when AI capabilities are relatively low, the marginal return from creative effort remains high compared to the mechanistic alternative. This encourages agents to use their creative margin to meet their performance goals. Because creative effort and inherent ability are technologically complementary, the usage of this creative margin restores the strict single-crossing property. High-type agents face a differentially lower cost of generating output on this dimension, ensuring that performance continues to function as a strictly monotonic signal of underlying ability, resulting in a separating equilibrium.

These theoretical results provide precise implications for labor markets under AI. The recent employer panic over AI implicitly assumes that large language models have permanently destroyed the single-crossing property by driving the cost of high-quality output to near-zero for all candidates.\footnote{The widespread diffusion of LLMs precipitated a crisis of confidence among human capital practitioners, who feared AI had commoditized knowledge work and collapsed separating equilibria into pooling ones. Because candidates could theoretically generate observationally identical deliverables regardless of ability, tech and legal firms increasingly abandoned asynchronous evaluations, such as take-home coding tests or entry-level legal briefs \citep{katz2024gpt}, in favor of live, proctored assessments.} Our model demonstrates that this fear is mathematically well-founded, but only for evaluations dominated by routine implementation. Ultimately, performance-based screening remains robust, provided that principals actively redesign evaluations to target the shifting frontier of human creativity.

As increasing AI capability crowds out creative effort, the exact design of the contest matters. By comparing the contest equilibrium to a single-agent baseline, we formally define \emph{performance manipulation} as the strategic over-investment in mechanistic effort induced by competitive pressure. We find a distinct threshold behavior: lower-type agents systematically engage in performance manipulation, flooding the principal with AI-augmented mechanistic output to artificially inflate their observable score. This dynamic yields a critical, actionable prescription for labor market screening: employers must fundamentally redesign evaluation exercises, such as take-home coding tests or legal research memos, to ensure they embed a dominant creative component. Tasks that can be solved primarily via mechanistic effort will inevitably pool candidates. Only tasks demanding structural synthesis, diagnostic problem-solving, or judgment under uncertainty will successfully separate them. Such shifts are already underway in practice. 
For example, Google is piloting new software engineering interviews in which candidates debug a broken codebase with AI assistants and are evaluated on their ability to prompt, interpret and verify AI-generated solutions.\footnote{\url{https://www.amigohelp.ai/blog/google-ai-assisted-coding-interview/}}
Furthermore, our model dictates that principals must implement more heavily skewed reward structures favoring top performers. Counterintuitively, as AI makes routine tasks easier, principals must steepen incentives to adequately compensate the high marginal cost of scarce creative effort and effectively restrain manipulation.

We provide empirical evidence on the effectiveness of performance-based screening using a unique and rich dataset from machine learning competitions on Kaggle. These competitions took place before the recent surge in generative AI following the release of ChatGPT in late 2022, corresponding to the low-AI-capability regime in our model. 
This rich dataset captures not only contestants’ performance but also their submitted code, allowing us to uncover the otherwise hidden problem-solving processes underlying observable performance metrics.
We leverage this unusually granular information to construct novel, direct measures of mechanistic and creative effort. Specifically, we construct these measures by prompting a large language model to identify mechanistic and creative techniques from contestants’ submitted code, and validate them through both human assessment and independent evaluations by alternative large language models.
Unlike conventional effort proxies, which are often indirect and heuristic, our measures are grounded in the concrete, step-by-step problem-solving process, making them more transparent, interpretable, substantively meaningful, and verifiable than black-box proxies.

Our regression analysis is consistent with the theoretical prediction that performance-based screening remains effective when AI capabilities are limited: high-type contestants achieve AUC scores roughly 4\% higher than low-type contestants on average, allowing underlying types to be distinguished based solely on observable performance.
In addition, we find a pronounced divergence in the effort composition of high- and low-type contestants: high types exert roughly 50\% more creative effort than low types, while their mechanistic effort is statistically indistinguishable. 
While AI commoditizes routine execution, this finding underscores the importance of incorporating creative problem-solving into evaluation exercises to restore the skill premium and effectively screen candidates in the AI-assisted era.


\subsection{Literature Review}

Performance manipulation is widely studied in AI. To our knowledge, however, most existing work focuses either on detecting \emph{reward hacking} by AI agents—the exploitation of misspecified objectives, also known as specification gaming or benchmark hacking \citep[e.g.,][]{deshpande2026benchmarking,wang2026broke}—or on designing models robust to strategic behavior by downstream users and data providers \citep{hardt2016strategic,bjorkegren2020manipulation,freeman2020no,ahmadi2021strategic,perdomo2020performative}. Our notion of performance manipulation is the human analog of these behaviors in a competitive labor-market setting: we share the concern with gamed objectives but study the incentives of human contestants in ML contests, rather than those of AI agents or downstream users.


As for modeling, our model relates to the classic literature on multi-tasking,
where agents allocate effort across multiple tasks
\citep[e.g.,][]{holmstrom1991multitask, benabou2016bonus}. Multi-tasking has been
studied in contexts such as firm management \citep{managerial1985}, academic
publishing \citep{ellison2002evolving}, and pharmaceutical innovation
\citep{bryan2022r}. Unlike this literature, which typically examines a single
agent deciding between tasks that all benefit the principal, we study a
multi-agent contest where only one type of effort benefits the principal.

Our model also relates to the broader literatures on contest theory and moral hazard in principal-agent theory \citep{clark1998competition,holmstrom1979moral,grossman1992analysis}, but differs by allowing agents to choose two distinct dimensions of effort rather than a single action.
Our division of effort between creation and mechanization is partly motivated by
the rise of generative AI and its rapid, large-scale deployment in the workplace
\citep[e.g.,][]{brynjolfsson2025generative,massenkoff2026labor}, including its
growing role in machine learning tasks \citep[e.g.,][]{chan2024mle}. However, this
is the first paper to formalize which ML techniques are readily mechanizable and
automatable by AI agents and which are not---an essential distinction for
understanding how AI tools are reshaping skill demands and competitive dynamics in
machine learning competitions.

Our paper also contributes to the literature on gaming behaviors.
\citet{phacking2020} studies ``p-hacking'' in data manipulation,\footnote{``p-hacking''
occurs when researchers selectively collect or analyze data until nonsignificant
results appear significant, representing mechanistic effort in our model.} while
\citet{ball2019scoring} and \citet{frankel2019muddled} examine gaming in signaling
games, and \citet{perez2022test} and \citet{li2023contests} analyze fraud-proof
information design and gaming in contests. In our setting, by contrast,
performance manipulation arises in machine learning contests as contestants
over-invest in mechanistic effort to inflate observable performance, splitting
effort between the creative margin the principal values and the mechanistic margin
AI has made cheap.

To the best of our knowledge, this paper is the first to develop a scalable
approach for measuring multi-dimensional effort from unstructured data. While a
large literature measures effort in laboratory and field experiments
\citep[e.g.,][]{gill2012structural,gill2019measuring,duflo2012incentives},
measuring effort in nonexperimental settings has been considerably more
challenging because process data that reveal effort (e.g., submitted documents or
prototypes) are often unavailable and, even when available, have been difficult to
quantify systematically and at scale. As a result, the few existing
nonexperimental studies have relied on coarse proxies for effort, such as working
hours
\citep[e.g.,][]{boudreau2011incentives}. Leveraging code
scripts from Kaggle competitions, we develop a novel, 
language-model-based approach to measuring effort from unstructured text. Our
approach also contributes to the growing literature on using AI and machine
learning to construct novel variables from unstructured data
\citep[e.g.,][]{caldara2022measuring,gorodnichenko2023voice,compiani2025demand,magnolfi2025triplet}.

\section{Model}\label{sec:model}

A principal (hereafter, ``she'') hosts a contest with $\numplayer$ players (each referred to as ``he''). 
Each player $i$ draws a private type $\type_i$ independently and identically from a commonly known distribution $\disttype$. This distribution has support $\Theta=[\underline{\type},\bar{\type}]$ and a bounded, atomless density function $\densitytype$. 
Types are private information; the principal and rival players know only the prior distribution $\disttype$. 
Let $\types=(\type_i)_{i=1}^{\numplayer}$ denote the type profile of all players, and $\typesminus=(\type_1,\dots,\type_{i-1},\type_{i+1},\dots,\type_{\numplayer})$ denote the profile of all players except $i$.  

Players compete based on their \emph{performance} in the contest. Denote player $i$'s performance by $\signal_i$. 
Each player's performance is a random draw from a continuous distribution $\distsignal_{\fitness_i}(\cdot)$ with \emph{mean performance} $\fitness_i$, support $\mathbb{R}_+$, and a bounded, atomless density function $\densitysignal_{\fitness_i}(\cdot)$. 
We assume the following condition on the performance distributions:

\begin{assumption}[FOSD]\label{a:fosd}
    The distribution $\distsignal_{\fitness}(\signal)$ first-order stochastically dominates $\distsignal_{\fitness'}(\signal)$ for any $\fitness>\fitness'$. That is, $\distsignal_{\fitness}(\signal)\leq\distsignal_{\fitness'}(\signal)$ for any $\signal\geq0$, with strict inequality for some $\signal\geq 0$.
\end{assumption}

A player determines his mean performance by allocating effort across two distinct dimensions: \emph{creative effort} $\mdeffort_i$ and \emph{mechanistic effort} $\pteffort_i$. The chosen effort profile determines the mean performance $\fitness_i$ via the production function:
\begin{align}\label{eq:fitness}
    \fitness_i = \mdeffortfn(\mdeffort_i, \type_i) + \pteffortfn(\pteffort_i; \gamma),
\end{align}
where $\gamma \in \mathbb{R}_+$ is an exogenous parameter capturing the prevailing technological frontier, which we interpret as the capability of AI. Effort choices are private information. Let $\fitnesses=(\fitness_i)_{i=1}^{\numplayer}$ denote the mean performance profile of all players.

We assume the creative production function $\mdeffortfn(\mdeffort_i, \type_i)$ is strictly increasing and strictly concave in $\mdeffort_i$, with $\mdeffortfn(0, \type_i) = 0$ for all $\type_i \in \Theta$. Crucially, we make the following assumption on $\mdeffortfn$:

\begin{assumption}[Supermodularity of $\mdeffortfn$]\label{a:sm}
    $\mdeffortfn(\mdeffort_i, \type_i)$ is strictly supermodular in $(\mdeffort_i, \type_i)$: $\frac{\partial^2\mdeffortfn(\mdeffort_i, \type_i)}{\partial\mdeffort_i\partial\type_i} > 0$.
\end{assumption}

In contrast, the mechanistic production function does not require private expertise, but its productivity is enhanced by the AI capability $\gamma$. We assume $\pteffortfn(\pteffort_i; \gamma)$ is strictly increasing and strictly concave in $\pteffort_i$, satisfying the Inada condition $\lim_{\pteffort \to \infty} \frac{\partial \pteffortfn(\pteffort;\gamma)}{\partial \pteffort} = 0$, with $\pteffortfn(0; \gamma) = 0$. Furthermore, AI capability universally raises baseline fitness: $\frac{\partial \pteffortfn(\pteffort_i; \gamma)}{\partial \gamma} > 0$. 
Moreover, AI capability universally raises the marginal productivity of mechanistic effort:

\begin{assumption}[Supermodularity of $\pteffortfn$]\label{a:sm-pt}
    $\pteffortfn(\pteffort_i; \gamma)$ is strictly supermodular in $(\pteffort_i, \gamma)$: $\frac{\partial^2\pteffortfn(\pteffort_i; \gamma)}{\partial\pteffort_i\partial\gamma} > 0$.
\end{assumption}

Throughout, we write $\pteffortfn'(\pteffort;\gamma)\coloneqq\frac{\partial \pteffortfn(\pteffort;\gamma)}{\partial \pteffort}$ and $\pteffortfn''(\pteffort;\gamma)\coloneqq\frac{\partial^2 \pteffortfn(\pteffort;\gamma)}{\partial \pteffort^2}$ for the partial derivatives of $\pteffortfn$ with respect to mechanistic effort, reserving $\frac{\partial \pteffortfn(\pteffort;\gamma)}{\partial \gamma}$ for the derivative with respect to the AI-capability parameter. 

\paragraph{Contest awards and payoffs.} 
The principal pre-specifies a prize vector $\rewards=(\reward_k)_{k=1}^{\numplayer}$ with $\reward_k\geq \reward_{k +1}\geq0$ for $1\leq k\leq \numplayer-1$, such that the player whose performance ranks $k$th receives the $k$th highest prize $\reward_k$. 
Denote the $k$th highest performance by $\signal_{(k)}$, such that $\signal_{(1)}\geq \signal_{(2)} \dots \geq \signal_{(\numplayer)}$. 
Let $i_k^*$ be the player whose performance ranks $k$th (i.e., $\signal_{i_k^*}=\signal_{(k)}$).

The total effort exerted by player $i$ is $\budget_i\coloneqq\mdeffort_i+\pteffort_i$. Let $\etcostfn(\budget_i)$ denote the cost of exerting this total effort.  
We assume $\etcostfn$ is twice differentiable, non-decreasing, and convex in $\budget_i$, with $\etcostfn(0)=0$.
Player $i$'s \emph{payoff} for exerting effort $(\mdeffort_i, \pteffort_i)$ and achieving a performance level $\signal_i$ is:
\begin{align}\label{eq:player's payoff}
\pi_i=\underbrace{\sum_{k=1}^\numplayer \indicate{i=i_k^*}\cdot \reward_k }_{\text{contest gain}} + \underbrace{\signal_i}_{\text{baseline gain}} - \etcostfn(\mdeffort_i+\pteffort_i),
\end{align}
where $\indicate{\cdot}$ is an indicator function equal to $1$ if player $i$ wins the $k$th highest prize. 
Players derive utility from two sources: \textit{contest gain} and \textit{baseline gain}. Contest gain represents the direct reward won from the competition—such as monetary prizes from a Kaggle contest, a workplace promotion, or grant funding in a research race. Baseline gain captures the benchmark return from absolute performance, such as an individual's enhanced reputation,\footnote{For example, winning or placing highly in a Kaggle competition increases a participant's visibility and signals their data science skills to potential employers. \citet{cisternas2018career} models these career concerns explicitly, where the market infers a player's type based on their observable signal.} or a baseline salary progression in the workplace. 

\paragraph{Timeline.} The game unfolds in three stages:   
\begin{enumerate}
    \item The principal announces the prize vector $\rewards$.
    \item Each player observes his own type $\type_i$ and chooses an effort allocation $(\mdeffort_i, \pteffort_i)$ to achieve mean performance $\fitness_i$.
    \item Performance levels are realized.\footnote{Because performance is drawn from a continuous distribution, ties occur with probability $0$.} Rewards are allocated according to the contest rules and final rankings. 
\end{enumerate}

\paragraph{Solution concept.}  Each player $i$'s strategy $\eqmfitness_i$ is a mapping from the type space $\Theta$ to an effort allocation $(\mdeffort_i,\pteffort_i)$.
We adopt the standard \emph{Bayesian Nash equilibrium} as our solution concept. 
A Bayesian Nash equilibrium is a strategy profile $\eqmfitnesses^*=(\eqmfitness_i^*)_{i=1}^{\numplayer}$ such that for each player $i$, $\eqmfitness_i^*$ maximizes his expected utility, given the strategies of all other players $\eqmfitnesses^*_{-i}=(\eqmfitness_j^*)_{j\neq i}$. Formally, $\eqmfitness_i^*$ solves the following optimization problem: 
\begin{equation}\tag{$\mathcal{P}_i$}\label{player i's problem}
    \begin{aligned}
        \max_{\eqmfitness_i(\type_i)} & \quad \mathbb{E}_{\typesminus, \signals}\left[\pi_i(\eqmfitness_i,\type_i, \eqmfitnesses^*_{-i})\mid \fitness(\eqmfitness_i;\type_i),\eqmfitnesses^*_{-i}\right]\\
    \text{s.t. } & \quad \fitness(\eqmfitness_i;\type_i)=\mdeffortfn(\mdeffort_i,\type_i)+\pteffortfn(\pteffort_i;\gamma).
    \end{aligned}
\end{equation}
Notably, this expectation is taken over both the private types of competing players and the random performance realization given the chosen mean performance $\fitness_i$. 

Let $\mdeffortcontest_i(\type_i)$, $\pteffortcontest_i(\type_i)$, and $\budgetcontest_i(\type_i)$ denote the equilibrium creative, mechanistic, and total effort, respectively, for player $i$ with type $\type_i$. Let player $i$'s equilibrium fitness be $\fitnesscontest_i(\type_i):= \fitness(\mdeffortcontest_i(\type_i),\pteffortcontest_i(\type_i);\type_i)$. Finally, a strategy profile $\eqmfitnesses^*$ is a \emph{symmetric} Nash equilibrium if $\eqmfitness_i^* = \eqmfitness_j^*$ for all $i, j$.   

\paragraph{Examples}

    \begin{example}[Machine Learning Contests]
        In machine learning (ML) contests, firms or academic institutions (principals) host competitions to elicit high-quality predictive models from the data science community. Notable examples include Kaggle, AIcrowd, the \emph{ImageNet Challenge} \citep{deng2009imagenet}, and the ongoing race in industrial LLM development. 
        
        Creative effort includes introducing a novel algorithmic framework, designing a new optimization strategy, or incorporating structural knowledge of a specific dataset. This delivers genuinely innovative solutions that depend heavily on the participant's private expertise ($\type_i$).

        Mechanistic effort includes deploying standard plug-and-play packages, executing brute-force hyperparameter searches, and applying default architectures. This trial-and-error optimization is highly substitutable by prevailing automated ML tools \citep{hutter2019automated} and AI capabilities ($\gamma$).\footnote{We provide detailed examples of creative and mechanistic effort from scripts submitted to Kaggle in \cref{tab:effort_classification_featureengineering} and Appendix~\ref{apx:technique_classification_examples}.}

        The contest gain is the monetary prize or formal recognition awarded to top-ranked models. The baseline gain captures the participant's increased visibility, enhanced portfolio, and labor market signaling \citep{cisternas2018career}, which scale with absolute performance regardless of final rank.
    \end{example}
    
    \begin{example}[Software Engineering Contracts]
       In the software labor market, freelance developers compete for lucrative client contracts \citep{agrawal2015digitization}, or internal engineers race for promotions and project leadership.
    
    Creative effort includes designing unprecedented system architectures, navigating ambiguous client requirements, and anticipating novel hardware bottlenecks. These tasks operate beyond documented solutions and rely on the developer's private expertise.
    
    Mechanistic effort includes writing boilerplate code,\footnote{Boilerplate code acts as a template that sets up the basic infrastructure required for a program to run.} translating logic across languages, and generating routine unit tests. These pattern-replicable tasks are directly enhanced by prevailing AI tools like GitHub Copilot, which substitute for routine human work \citep{peng2023impact}.
    
    The contest gain consists of securing the freelance contract, the internal promotion, or project leadership. The baseline gain captures the general expansion of the developer's portfolio, skill improvement, and baseline salary progression, which accrue based on absolute code quality.
    \end{example}
    
    \begin{example}[Legal Practice and Client Pitches]
        Law firms frequently pitch against rival practices to secure limited corporate clients, acting as players in a ranked contest where the corporate client is the principal \citep{galanter1991tournament}.
    
    Creative effort includes formulating novel litigation strategies, assessing opponent psychology, and structuring first-of-its-kind M\&A deals in emerging regulatory environments. Such tasks demand high-level strategic judgment that strongly complements the lawyer's private type.
    
    Mechanistic effort includes conducting e-discovery, drafting boilerplate contracts (e.g., NDAs), and verifying legal citations. These tasks depend on retrieving and summarizing historical precedent, making their marginal productivity highly responsive to advances in AI capability \citep{remus2017can}.
    
    The contest gain represents acquiring the lucrative client or securing a favorable trial verdict. The baseline gain reflects the firm's broader reputation building, billable hours generated during research, and the long-term enhancement of the attorney's credibility.
    \end{example} 
    
    \begin{example}[Scientific Research and Grant Funding]
        Research laboratories compete for limited grant funding (e.g., from the NSF or NIH) \citep{stephan1996economics}, or race to publish breakthrough discoveries in top journals. Here, the grant agency or editorial board acts as the principal.
    
    Creative effort includes generating novel hypotheses, designing new identification strategies, and relying on intuition to process negative results or data anomalies. This fundamentally pushes the knowledge frontier and depends on the researcher's domain-specific expertise.
    
    Mechanistic effort includes surveying existing literature, executing pre-specified statistical regressions, and formatting manuscripts. These routine, rule-based tasks are easily substituted and enhanced by LLMs trained on historical scientific corpora \citep{korinek2023language}.
    
    The contest gain is the discrete reward of securing the grant funding or publishing in a prestigious journal. The baseline gain captures the underlying academic value generated, such as incremental citations, generalized knowledge accumulation, and steady progress toward academic tenure.
    \end{example}

\section{Equilibrium Existence}

Observe that in player $i$'s problem (\ref{player i's problem}), the effort allocated across the two actions determines the agent's utility through two channels: (1) the mean performance $\fitness_i$, which dictates the resultant probability distribution over the contest rankings, and (2) the total effort exerted, which determines the cost. We can transform each player's strategy space into an alternative, equivalent formulation. In this auxiliary game, each agent first chooses a target mean performance $\fitness_i$ to maximize his utility, and then chooses the least-costly allocation of effort across the two actions to achieve that target $\fitness_i$. This two-stage perspective allows us to disentangle the effects of performance and effort allocation, significantly simplifying the equilibrium analysis. 

To describe this auxiliary game, we first define the minimum cost for a player of type $\type_i$ to achieve mean performance $\fitness_i$ under AI capability $\gamma$ as the following function: 
\begin{equation}\label{cost function of fitness}
    \fscostfn(\fitness_i,\type_i;\gamma)=\min_{\mdeffort_i,\pteffort_i\geq 0}\etcostfn(\mdeffort_i+\pteffort_i) \qquad \text{   s.t. }  \, \, \mdeffortfn(\mdeffort_i,\type_i)+\pteffortfn(\pteffort_i;\gamma)=\fitness_i.
\end{equation}

In the auxiliary game, when player $i$ has type $\type_i$ and chooses mean performance $\fitness_i$, his expected utility is given by $\sum_{k=1}^\numplayer \indicate{i=i_k^*}\cdot \reward_k+\fitness_i-\fscostfn(\fitness_i,\type_i;\gamma)$. 
Agent $i$'s strategy in this alternative game is characterized by the pair $(\eqmfitnessalt_i, \eqmeffortalts)$, where $\eqmfitnessalt_i$ is a mapping from his type space $\Theta$ to a choice of mean performance $\fitness_i$, and $\eqmeffortalts=(\eqmeffortalt_1,\eqmeffortalt_2)$ maps the type and mean performance space $\Theta\times\mathbb{R}_+$ to the optimal allocation of the two types of effort. 

The following lemma formally establishes the connection between our original game and the auxiliary game. Armed with this equivalence, we conduct the remainder of our analysis on the auxiliary game.

\begin{lemma}\label{lemma:equivalent}
    There is a one-to-one mapping from any equilibrium $(\eqmfitness_i^*)_i$ of the original game to an equilibrium $((\eqmfitnessalt_i^\ddag)_i, \eqmeffortalts^\ddag)$ in the auxiliary game, such that $\eqmfitness_i^*(\type)=\eqmeffortalts^\ddag(\eqmfitnessalt_i^\ddag(\type),\type)$ for any $\type\in\Theta$.
\end{lemma}

The intuition behind this equivalence is that in the original game, each player's payoff depends on the strategies of the other players only through their mean performance. In other words, the other players' mean performance profile acts as a sufficient statistic for their strategies. The formal proof is provided in Appendix~\ref{apx:auxiliary_game}.

\begin{proposition}\label{prop:exist}
    Under Assumption \ref{a:fosd} and \ref{a:sm}, there exists a pure strategy Nash equilibrium that is: 
    \begin{itemize}
        \item \emph{symmetric}:  $\eqmfitnessalt_i^*(\type_i)=\eqmfitnessalt^*(\type_i)$ for all $i$; 
        \item and \emph{monotone}: $\eqmfitnessalt^*(\type_i)$ is non-decreasing in $\type_i$. Moreover, when the derived cost function $\fscostfn(\fitness, \type; \gamma)$ satisfies strict decreasing differences in $(\fitness,\type)$, $\eqmfitnessalt^*(\type)$ is strictly increasing. 
    \end{itemize}
\end{proposition} 

Our proof, detailed in Appendix~\ref{apx:existence_proof}, demonstrates that the payoff function in the auxiliary game is bounded, continuous in $\fitness_i$, and exhibits increasing differences in both $(\fitness_i,\type_i)$ and $(\fitness_i,\gamma)$. This allows us to evoke Theorem 4.5 of \citet{reny2011existence} to establish the existence of a symmetric monotone equilibrium. 

Intuitively, the monotonicity arises from the strategic complementarity between creative effort and the player's inherent type (\cref{a:sm}). High-type agents hold a comparative advantage in generating creative output, which makes achieving higher aggregate performance differentially less costly for them. When this cost advantage is strict, it naturally separates the types, as high-ability players will strictly prefer to build better-performing model to secure higher contest rankings. 

From this point forward, we focus on the symmetric monotone equilibrium and drop the subscript $i$ when it is clear from the context. To understand how the prevailing technology frontier shapes this equilibrium, we introduce a boundary condition on the capability of AI.

\begin{assumption}\label{a:gamma_limits}
    The marginal productivity of mechanistic effort spans the extremes as AI capability $\gamma$ varies: 
    $\lim_{\gamma \to 0} \pteffortfn'(0; \gamma) = 0$, and for any finite $\pteffort$, $\lim_{\gamma \to \infty} \pteffortfn'(\pteffort; \gamma) = \infty$.
\end{assumption}

\begin{proposition}[AI Capabilities and Equilibrium Regimes]\label{prop:ai_regimes}
    Under Assumptions \ref{a:fosd}, \ref{a:sm}, \ref{a:sm-pt}, and \ref{a:gamma_limits}, the equilibrium regime is dictated by the AI capability relative to the bounds of the type distribution $\Theta = [\underline{\type}, \bar{\type}]$:
    \begin{enumerate}
        \item \textbf{Strict Separation (Limited AI):} There exists a lower threshold $\underline{\gamma} > 0$, determined by the lower bound of the type space $\underline{\type}$, such that for all $\gamma < \underline{\gamma}$, the derived cost function $\fscostfn(\fitness, \type; \gamma)$ exhibits \emph{strict} decreasing differences in $(\fitness, \type)$. Consequently, any interior symmetric equilibrium is strictly monotone.
        \item \textbf{Complete Pooling (Advanced AI):} There exists an upper threshold $\bar{\gamma} > 0$, determined by the upper bound of the type space $\bar{\type}$, such that for all $\gamma > \bar{\gamma}$, the derived cost function exhibits only \emph{weak} decreasing differences. Consequently, the strict single-crossing property fails, and the equilibrium is pooled.
    \end{enumerate}
\end{proposition}

The proof is provided in Appendix~\ref{apx:existence_proof}.
The intuition for  \cref{prop:ai_regimes} lies in the substitutability of effort. When AI is weak ($\gamma < \underline{\gamma}$), all agents optimally utilize at least some creative effort, preserving the cost advantage of high-type agents and ensuring strict separation. However, when AI is sufficiently advanced ($\gamma > \bar{\gamma}$), mechanistic effort becomes so overwhelmingly productive that it crowds out creative effort entirely—even for the most capable agents. Because mechanistic effort does not rely on private expertise, the cost of achieving any performance level becomes identical across all types, destroying the single-crossing property and forcing a pooled equilibrium. When AI capability is moderate, separation and pooling coexist: low types pool, while high types separate.

\subsection{Implications of AI on Screening}

\paragraph{Implications for Screening in the AI Era.}
The existence of a monotone equilibrium, where higher types consistently achieve better mean performance, provides a formal resolution to the widespread anxiety surrounding performance-based screening in the presence of generative AI.

Following the public release of advanced large language models, a prominent narrative emerged among human capital practitioners that AI had commoditized knowledge work, effectively collapsing a separating equilibrium into a pooling one. Because low-ability and high-ability candidates can produce observationally identical deliverables on standard coding or writing assessments through differential reliance on AI, employers have increasingly abandoned asynchronous evaluations, such as the take-home test, out of fear that output quality no longer signals underlying ability.

Our theoretical framework demonstrates that these fears are mathematically well-founded, but only if the evaluation mechanism remains static. In the language of \citet{spence1973}, a signaling mechanism successfully sorts candidates only if the marginal cost of producing the signal is strictly lower for high-type agents—the strict single-crossing property. As established in Proposition \ref{prop:ai_regimes}, when AI capabilities ($\gamma$) become sufficiently advanced ($\gamma > \bar{\gamma}$), they act as a universal, highly productive substitute for creative effort. This compresses the distribution of outputs and crowds out creative effort entirely. Because all agents optimally rely on the AI-enhanced mechanistic margin, the cost of achieving any performance level becomes independent of inherent type. The strict decreasing differences in the cost function vanish, plunging the market into a complete pooling equilibrium.

However, our model also clarifies that AI does not universally destroy the capacity for screening; rather, it changes the dimension along which stratification occurs. The quality of output depends on both mechanistic and creative effort. Because creative effort and type are strategic complements (Assumption \ref{a:sm}), the marginal return to creative effort strictly increases with the agent's inherent ability.

\paragraph{Practical Implications for Task Design.}
The economic return to mechanistic proficiency is rapidly being competed away by rising $\gamma$. Consequently, the principal's problem is not to abandon performance-based screening, but to redesign evaluations to actively restore the strict increasing differences in the cost function.

An evaluation composed entirely of routine implementation—where the mechanistic component dominates—inevitably falls into the pooling regime ($\gamma > \bar{\gamma}$) because it can be completed by submitting a standard prompt to an LLM. To distinguish candidates, evaluations must be restructured to raise the marginal return of the creative production function $\mdeffortfn$, effectively pushing the upper threshold $\bar{\gamma}$ higher and forcing the evaluation back into the separating regime ($\gamma < \underline{\gamma}$).

By requiring candidates to diagnose unfamiliar systems, propose non-standard architectural solutions, or reason explicitly about trade-offs under uncertainty, employers shift the production weight onto the creative margin. On this margin, the single-crossing property is resurrected through the supermodularity of quality. High-type agents will again deliver better-performing outputs because it is differentially less costly for them, allowing principals to distinguish types based on aggregate performance. The skill premium does not disappear; it migrates to the capacity for problem formulation, strategic judgment, and non-routine synthesis.

\paragraph{Dynamic Human Capital and the Race Against the Pooling Threshold.}

Our baseline model treats the distribution of human capability, $\disttype$, as exogenous and fixed. In the long run, however, AI may reshape human capital accumulation process itself. Growing empirical evidence suggests that generative AI can accelerate skill acquisition rather than merely substitute for human effort. AI assistance helps novice workers learn the practices of high performers and move faster along the experience curve \citep{brynjolfsson2023generative}, disproportionately improves the performance of lower-skilled knowledge workers \citep{dellacqua2023navigating}, and, when designed with pedagogical guardrails, can enhance skill retention and independent performance \citep{bastani2024generative}. Thus, widespread AI adoption may endogenously reshape the distribution of human capability over time.



These empirical findings carry a profound implication for our equilibrium regimes. If interacting with AI systematically improves human capability over time by shifting the type distribution $\disttype$ to the right and raising both $\underline{\type}$ and $\bar{\type}$, then the thresholds for the equilibrium regimes defined in Proposition \ref{prop:ai_regimes} are not static. 

Because the pooling threshold $\bar{\gamma}$ is strictly increasing in the upper bound of human capability ($\bar{\type}$), AI-driven upskilling inherently pushes $\bar{\gamma}$ higher. Therefore, for generative AI to permanently collapse the labor market into a complete pooling equilibrium, AI capability ($\gamma$) must not only grow, but it must grow faster than its own accelerating effect on human learning. The long-run equilibrium is therefore characterized by a dynamic race: as AI capabilities advance and threaten to crowd out the creative margin, they simultaneously raise the ceiling of human creative capabilities, pushing the pooling threshold further out and continually preserving the space for a separating equilibrium.

However, overly relying on AI may instead erode human knowledge and skills over time. For instance, students' exam scores have been found to decline substantially following AI adoption, potentially because they outsource learning tasks such as homework to AI \citep{stromberg2026generative}. Without deliberate interventions to preserve and cultivate human expertise, the distribution of human capability may therefore shift leftward over time, accelerating collapse of the labor market into a pooling equilibrium. We leave the dynamic interplay between AI capabilities and human capital accumulation for future work.

\section{Performance Manipulation}

In this section, we explore a fundamental principal-agent conflict: the principal seeks to elicit true innovation driven by creative effort, but can only contract upon an observable performance metric that is susceptible to inflation via mechanistic effort. This agency friction is dramatically amplified in the era of generative AI, as the precipitous drop in the cost of algorithmic substitution makes mechanistic effort both cheap and abundant \citep{makingtalkcheap}. Because organizations must frequently rely on imperfect performance evaluations to incentivize complex, multi-dimensional tasks \citep{holmstrom1991multitask}, this asymmetric augmentation of productivity poses a severe risk of performance manipulation. 

Our primary objective is to formally define this phenomenon, characterize the types of agents most prone to engage in it, and propose structural mitigations. To isolate the effect of competition, we first introduce a single-player baseline model in which an agent allocates effort solely to maximize his baseline payoff. We then compare this baseline outcome with the equilibrium effort allocation in a competitive contest. Building on this comparative static, we formally define performance manipulation as the strategic substitution toward mechanistic effort induced by competition. We conclude by demonstrating how optimal contest design can mitigate this misalignment. 

\subsection{The Baseline} 

We first establish a single-agent baseline scenario where a player's payoff is derived exclusively from the baseline gain, absent any contest incentives. Suppressing the player index $i$, the agent solves:
\begin{align}\tag{$\mathcal{P}_b$}\label{benchmark}
    \max_{\mdeffort, \pteffort \geq 0} \quad \payoffbm(\mdeffort, \pteffort; \type, \gamma) = \fitness(\mdeffort, \pteffort; \type, \gamma) - \fscostfn(\fitness(\mdeffort, \pteffort; \type, \gamma), \type; \gamma).
\end{align}
Let $(\mdeffortbenchmark(\type), \pteffortbenchmark(\type)) = \argmax_{\mdeffort,\pteffort\geq 0}\payoffbm(\mdeffort,\pteffort;\type,\gamma)$ denote the optimal effort allocation, and $\budgetbenchmark(\type) = \mdeffortbenchmark(\type) + \pteffortbenchmark(\type)$ denote the optimal total effort in the baseline scenario. To ensure interior total effort, we impose standard boundary conditions on the marginal returns at the origin.

\begin{assumption}\label{a:pt}
    For the relevant range of $\gamma$, $\pteffortfn'(0; \gamma) > \etcostfn'(0)$. Furthermore, $\lim_{\type \to \underline{\type}} \frac{\partial\mdeffortfn(0, \type)}{\partial\mdeffort} < \etcostfn'(0)$ and $\lim_{\type \to \bar{\type}} \frac{\partial\mdeffortfn(0, \type)}{\partial\mdeffort} > \pteffortfn'(0; \gamma)$.
\end{assumption}

Assumption \ref{a:pt} ensures that all types exert strictly positive total effort ($\budgetbenchmark(\type) > 0$), and that the extreme types sort into corner solutions regarding their effort composition.

\begin{proposition}[Baseline Effort Allocation]\label{prop:eqm_benchmark}
    Under Assumptions \ref{a:fosd}, \ref{a:sm}, and \ref{a:pt}, there exist two unique thresholds $\typebl < \typebu$ in $\Theta$ such that:
    \begin{enumerate}
        \item \textbf{Pure Mechanization:} Types $\type \leq \typebl$ allocate all effort mechanistically ($\mdeffortbenchmark = 0, \pteffortbenchmark > 0$).
        \item \textbf{Interior Allocation:} Types $\type \in (\typebl, \typebu)$ utilize both efforts ($\mdeffortbenchmark > 0, \pteffortbenchmark > 0$).
        \item \textbf{Pure Creation:} Types $\type \geq \typebu$ allocate all effort creatively ($\mdeffortbenchmark > 0, \pteffortbenchmark = 0$).
    \end{enumerate}
\end{proposition}

\begin{figure}[htp]
    \captionsetup{font=small}
    \caption{\textbf{Graphical illustration of \cref{prop:eqm_benchmark}}}
    \label{fig:eqm_benchmark}
    \centering
    \begin{tikzpicture}
\begin{axis}[
    width=11cm,
    height=8cm,
    axis lines=left,
    xmin=0, xmax=8.4,
    ymin=0, ymax=6.4,
    xtick=\empty,
    ytick=\empty,
    xlabel={},
    ylabel={},
    clip=false,
    axis line style={->},
]

\def\bdag{1.8}
\def\adag{2.8}
\def\edag{4.90}

\addplot[
    domain=0:8,
    samples=200,
    very thick,
    gray
] {0.75 + 0.085*x^2};

\node[gray!80!black] at (axis cs:7.0,5.85)
    {$c'(e)$};

\addplot[
    domain=1.05:8,
    samples=200,
    very thick,
    blue
] {5.25*exp(-0.23*x)};

\node[blue] at (axis cs:7.25,1.55)
    {$\displaystyle \frac{\partial v(a,\theta)}{\partial a}$};

\addplot[
    domain=0:8,
    samples=200,
    very thick,
    red
] {4.15*exp(-0.24*x)};

\node[red] at (axis cs:7.2,0.4)
    {$\xi'(b,\gamma)$};


\addplot[
    domain=0:1.05,
    samples=50,
    very thick,
    blue!80!black
] {5.1 - 0.9*x};

\addplot[
    domain=0:1.05,
    samples=50,
    very thick,
    green!80!black
] {5.1 - 0.9*x};

\addplot[
    domain=1.05:8,
    samples=150,
    very thick,
    green!80!black
] {4.18*exp(-0.105*(x-1.05))};


\draw[dashed]
    (axis cs:0,4.15) --
    (axis cs:1.05,4.15);

\draw[dashed]
    (axis cs:0,2.75) --
    (axis cs:\edag,2.75);

\draw[dashed]
    (axis cs:\bdag,0) --
    (axis cs:\bdag,2.75);

\draw[dashed]
    (axis cs:\adag,0) --
    (axis cs:\adag,2.75);

\draw[dashed]
    (axis cs:\edag,0) --
    (axis cs:\edag,2.75);

\node[below] at (axis cs:\bdag,0)
    {$b^\dagger$};

\node[below] at (axis cs:\adag,0)
    {$a^\dagger$};

\node[below] at (axis cs:\edag,0)
    {$e^\dagger$};

\node[right] at (axis cs:8.4,0)
    {$e,a,b$};

\end{axis}
\end{tikzpicture}
\end{figure}

\cref{fig:eqm_benchmark} provides graphical intuition for \cref{prop:eqm_benchmark}. 
The blue and red curves represent the marginal productivity of creation and mechanization, respectively. The green curve is the horizontal sum of the blue and red curves, representing the marginal productivity of total effort $\budget$. 

An optimizing agent equates the marginal productivities of the active effort margins. For sufficiently high (low) types, the initial marginal productivity of creation (the blue curve) lies strictly above (below) that of mechanization (the red curve). The kink in the green curve represents the threshold where the agent switches from utilizing a single effort margin to utilizing both. When $\type$ is sufficiently high (low), the green curve intersects the marginal cost curve $\etcostfn'(\budget)$ to the left of the kink, resulting in pure creation (pure mechanization). For intermediate types, the intersection occurs to the right of the kink, strictly positive effort is allocated to both margins.

Next, we examine how an advancement in AI capability impacts these effort allocations. Since an increase in AI capability ($\gamma$) raises the marginal return to mechanistic effort, it alters the opportunity cost of creation. The following corollary summarizes this structural shift: absent competition, higher AI capabilities shrink the set of types that exert positive creative effort.

\begin{corollary}[AI Capability and Effort Substitution]\label{coro:thresholds_gamma}
    The thresholds $\typebl(\gamma)$ and $\typebu(\gamma)$ are strictly increasing in AI capability $\gamma$. Consequently, an exogenous increase in $\gamma$ expands the measure of types engaged in pure mechanization, $[\underline{\type}, \typebl)$, and strictly shrinks the measure of types engaged in pure creation, $(\typebu, \bar{\type}]$.
\end{corollary}

\subsection{Contest Rewards and Incentives}

\begin{definition}[A partial order on prize vectors]\label{def:compare vectors}
 Consider two prize vectors  $\rewards,\rewards'\in\mathbb{R}_+^\numplayer$. We say $\rewards \geqq \rewards'$ if 
 $\reward_k-\reward_{k+1}\geq\reward'_k-\reward'_{k+1}$ for any $k=1,...,\numplayer-1$. 
\end{definition}

This partial order compares the \emph{skewness} of two prize vectors. The prize vector $\rewards$ is considered more skewed toward top-ranked contestants than $\rewards'$ if the incremental increase in prize reward for each step-up in rank is greater under $\rewards$. For example, if $\numplayer=4$, then $(6,4,3,0)\geqq(6,5,4,3)$.

We focus on the symmetric monotone pure strategy equilibrium, $\eqmfitnessalt^*(\type,\rewards)$.\footnote{Here, ``monotone" indicates that $\eqmfitnessalt^*(\type,\rewards)$ is non-decreasing in $\type$, but not necessarily in $\rewards$.} The following main result demonstrates that an agent of any type chooses to produce a better-fitted model under a more skewed prize vector. 

\begin{proposition}\label{monotone strategy in reward}
    For any $\type\in\Theta$, and any $\rewards\geqq\rewards'$, 
    $\eqmfitnessalt^*(\type,\rewards)\geq\eqmfitnessalt^*(\type,\rewards')$.
\end{proposition}

The proof follows a logic similar to the equilibrium existence argument: instead of showing that the payoff function has increasing differences in $(\fitness, \type)$, we show it has increasing differences in $(\fitness, \rewards)$. This increasing difference arises because (1) the rank distribution resulting from a higher mean performance first-order stochastically dominates that of a lower mean performance, and (2) a more skewed prize vector allocates disproportionally greater rewards to top performers. Intuitively, the marginal return on effort from increasing prize skewness is greater when a player is more likely to achieve a high rank. Thus, a more skewed prize vector incentivizes all types to target a higher mean performance.

With this established, we can compare equilibrium mean performance in the contest with the baseline model.

\begin{corollary}[Contest versus Baseline]\label{prop:benchmark and contest eqmfitness}
    Under \cref{a:fosd}, \ref{a:sm}, and \ref{a:pt}, $\fitnesscontest(\type)>\fitnessbenchmark(\type)$ for any $\type\in\Theta$.
\end{corollary}

\cref{prop:benchmark and contest eqmfitness} is a direct implication of \cref{monotone strategy in reward}. The baseline scenario is equivalent to a contest where the prize vector is zero ($\boldsymbol{0}=(0,...,0)$). By \cref{def:compare vectors}, any valid prize vector $\rewards$ satisfies $\rewards\geqq\boldsymbol{0}$. Thus, equilibrium mean performance in any active contest strictly dominates the baseline scenario.

\begin{corollary}[Equilibrium mean performance and AI capability]\label{coro:fitness_gamma}
    The symmetric monotone equilibrium mean performance $\fitnesscontest(\type)$ is non-decreasing in the AI capability $\gamma$ for every $\type\in\Theta$.
\end{corollary}

The proof of \cref{coro:fitness_gamma} mirrors that of \cref{monotone strategy in reward}. By Property 7 of \cref{lemma:budget_property}, the auxiliary-game payoff exhibits increasing differences in $(\fitness,\gamma)$. A higher $\gamma$ uniformly lowers the marginal cost of mean performance, inducing every type to target a higher performance level. \cref{coro:fitness_gamma} implies that an exogenous increase in AI capability naturally causes inflation in overall mean performance.

\subsection{Performance Manipulation}

Consider a symmetric monotone equilibrium $(\mdeffortcontest(\type),\pteffortcontest(\type))$ of the ML contest. Type $\type$ is involved in \emph{performance manipulation} if $\mdeffortcontest(\type)\leq\mdeffortbenchmark(\type)$ and $\pteffortcontest(\type)>\pteffortbenchmark(\type)$, where $\mdeffortbenchmark(\type)$ and $\pteffortbenchmark(\type)$ represent creative and mechanistic effort in the baseline setting. In other words, a contestant engages in performance manipulation if contest incentives induce a bias specifically toward mechanistic effort.

The following proposition compares the allocation of effort across the baseline and contest settings. 
 
\begin{proposition}\label{prop:compare}
Under \cref{a:fosd}, \ref{a:sm}, and \ref{a:pt}, there exists a threshold $\typecl\in (-\infty,\typebl]$ such that:
    \begin{enumerate}
        \item Any type $\type<\typecl$ only mechanizes, and $\pteffortcontest(\type)>\pteffortbenchmark(\type)>0$.
        \item Any type $\type\in(\typecl,\typebl)$ mechanizes in both settings with $\pteffortcontest(\type)>\pteffortbenchmark(\type)>0$, but creates only in the contest.
        \item Any type $\type\in (\typebl,\typebu)$ both creates and mechanizes, and each effort is strictly higher in the contest.
    \end{enumerate}
\end{proposition}

This proposition demonstrates that in the symmetric monotone contest equilibrium, \emph{low} types achieve higher mean performance solely by increasing mechanistic effort, whereas \emph{high} types achieve it by increasing effort across both margins.

The baseline thresholds $\typebl$ and $\typebu$ are non-decreasing in $\gamma$ (\cref{coro:thresholds_gamma}). The contest threshold $\typecl$, however, responds to $\gamma$ through two opposing channels. Holding equilibrium mean performance fixed, a higher $\gamma$ raises the marginal product of mechanistic effort, making pure mechanization attractive for a wider range of types and pushing $\typecl$ up. Simultaneously, a higher $\gamma$ raises the equilibrium mean performance each type targets (\cref{coro:fitness_gamma}), which requires more total mechanistic effort. Due to concavity, this lowers its marginal product and pushes $\typecl$ down. While the net effect on $\typecl$ depends on the strength of contest competition relative to the curvature of $\pteffortfn$, the comparative statics across regimes remain robust: relative to the baseline, low types systematically bias toward mechanistic effort.

\begin{corollary}\label{low overfit}
    Low types engage in performance manipulation, while high types do not.
\end{corollary}

This corollary follows immediately from \cref{prop:compare}. Any type $\type\leq\typecl$ exerts zero creative effort in both the contest and baseline scenarios, but strictly higher mechanistic effort in the contest. In contrast, high types do not engage in performance manipulation because they strictly increase their creative effort alongside their mechanistic effort in a contest. In the baseline scenario, the highest types ($\type \geq \typebu$) rely entirely on their private expertise, finding it optimal to exert zero mechanistic effort. However, in a highly skewed contest, the required equilibrium mean performance $\fitnesscontest(\type)$ is pushed significantly higher. Because the creative production function $\mdeffortfn$ is strictly concave, the marginal product of creative effort steadily diminishes as the high-type agent scales their output to remain competitive. 
If the contest incentives are sufficiently steep, this marginal product is driven down until it binds against the initial marginal product of mechanistic tools ($\pteffortfn'(0; \gamma)$). Consequently, the pure-creation regime can completely unravel: to maintain their rank, even the absolute highest-ability agents are eventually forced to supplement their domain expertise with routine algorithmic execution.
In Appendix \ref{append:examples}, we provide concrete functional forms to illustrate how types in each region allocate their effort.

\paragraph{Real-World Implications: Proxy Misalignment and Benchmark Hacking.} 
We define this competition-induced effort shift as performance manipulation. Because effort is unobservable, the performance metric is only an imperfect proxy for the principal's true objective, creating the classic organizational pathology of “rewarding A while hoping for B” \citep{kerr1975folly}. The principal hopes to elicit agents' private expertise ($\type_i$) through creative problem-solving ($\mdeffort_i$), but the contest strictly rewards observable performance ($\fitness_i$), potentially diverting effort from goals such as long-term software maintainability, out-of-sample model robustness, or genuine legal synthesis \citep{holmstrom1991multitask}. Competition further exacerbates this misalignment by incentivizing agents to inflate the observable metric. We focus on this competition-induced distortion because, unlike the underlying misalignment caused by imperfect performance measurement, which is often inherent and difficult to eliminate without advances in measurement technology, it can be effectively mitigated through contest design.


In the presence of highly capable AI ($\gamma$), the cost of inflating this proxy via mechanistic effort ($\pteffort_i$) plummets. When an agent substitutes toward mechanistic effort, they explicitly exploit this proxy misalignment, artificially inflating the observable metric using readily available technology rather than supplying the scarce human judgment the principal seeks. This provides a formal microfoundation for Goodhart's Law in the AI era: as generative models drastically lower the cost of mechanistic output, standard performance metrics become increasingly susceptible to manipulation.

In machine learning contests specifically, this manifests as \emph{benchmark hacking}: the practice of optimizing only for contest rules through trial-and-error approaches without faithfully solving the underlying problem. From the perspective of the contest host, the goal is often to elicit a model that delivers meaningful innovation capable of generalizing to future problems.\footnote{While some Kaggle contests prioritize solving a one-off problem without concern for generalizability, this paper focuses on contests where the true goal is to leverage the winning model for similar future tasks or to advance research within a broader community.} However, because prizes are awarded based on model performance on a static dataset, contestants are incentivized to aggressively tune hyperparameters and make numerous submissions to find the configuration that yields the highest local score. Unlike a single data analyst developing a model in isolation (our baseline), ML contests introduce a severe incentive misalignment. As long as the principal cannot perfectly distinguish genuine innovation from AI-assisted benchmark hacking—a difficult task given the cost of manually reviewing thousands of submissions—this structural manipulation persists.

\subsection{Implications for Contest Design}

The following corollary follows immediately from Propositions \ref{monotone strategy in reward} and \ref{prop:compare}.

\begin{corollary}\label{coro:actionble policy}
    For any two prize vectors such that $\rewards\geqq\rewards'$: 
    \begin{enumerate}
        \item It is never optimal to reward the lowest-ranked player.
        \item A strictly smaller measure of types engages in metric manipulation under the more skewed prize vector $\rewards$.
        \item The more skewed prize vector $\rewards$ induces weakly higher mean performance for every type.
    \end{enumerate}
\end{corollary}

This corollary yields a striking, and seemingly counterintuitive, implication for contest design in the AI era. As AI capabilities advance and mechanistic effort becomes cheaper and more readily substitutable, one might expect that principals could lower compensation because the overall task is easier. Our model demonstrates the exact opposite: the principal must paradoxically implement more lucrative and heavily skewed prize structures to sustain genuine innovation. 

This aligns closely with the multi-task agency framework of \citet{holmstrom1991multitask}. As the mechanistic dimension of performance becomes drastically cheaper to execute, it threatens to completely crowd out the more valuable, yet harder-to-generate, creative dimension. To prevent participants from substituting toward cheap algorithmic output, the principal must steepen the incentive structure to sufficiently compensate high-type agents for the steep marginal cost of creative effort. 

Consequently, a more skewed prize vector serves as a structural mitigation against metric manipulation, yielding three distinct benefits: it (1) elicits higher aggregate mean performance (as established in \cref{monotone strategy in reward}), (2) stimulates greater creative effort across the participant pool, and (3) systematically reduces the mass of types engaged in metric manipulation. While this establishes the directional necessity of steeper incentives, the exact optimal tradeoff between the principal's total prize budget and expected performance gains remains an open question. We leave a formal characterization of the optimal, profit-maximizing reward structure to future work.

\section{Empirical Validation with Kaggle Data}
\label{sec:fact}

In this section, we validate our modeling choices and theoretical predictions using public data on featured competitions hosted by a popular ML contest platform \emph{Kaggle} \citep{lemus2021dynamic}.\footnote{Featured competitions are the types of competitions in which players build a machine learning model to solve a (usually commercially-purposed) prediction task. 
Unlike other types of competitions on Kaggle, featured competitions offer medium- to high-stake cash prizes that can go as high as million dollars.
}  
In a typical ML competition, there is a prediction task and players compete on the performance of their designed ML models for it.
They are provided with a labeled \emph{training} dataset and an unlabeled \emph{test} dataset. 
Each participant trains an ML model on the training data, uses it to predict labels for the test data, and submits the resulting predictions for feedback.
The platform evaluates each submission against the hidden test labels but reveals performance only on a predetermined and undisclosed subset of the test data (typically 70\%) to mitigate overfitting. Based on this feedback, players may refine their models and submit updated predictions throughout the competition. At the end of the contest, each player makes a final submission. The performance on the remaining withheld portion of the test data (typically 30\%) determines the rankings and the winner(s).


\paragraph{Data. } To ensure comparability across contests, we restrict our sample to competitions involving binary classification tasks with prize amounts in U.S. dollars. 
The evaluation metric for these competitions is the \emph{Area Under the Receiver Operating Characteristic Curve} (AUC), which ranges from 0\% to 100\%, with higher values indicating better model performance.
Eventually, our dataset contains 28 contests. The unit of observation is a contest-player pair, with contest indexed by $j$ and player by $i$.  
For each contest, the variables we observe include number of players (usually around 1000$\sim$3000), number of prizes (usually around 1$\sim$4), total reward quantity (usually within 10K$\sim$70K USD), and contest year (within 2010$\sim$2021), etc. 
At the contest-player level, we observe each player's final performance (measured by AUC, which typically ranges between 60\% and 95\%) and whether the player has publicly shared their code. In total, we collect approximately 1,500 publicly shared code scripts, which we later use to measure players' effort choices (creative vs. mechanistic).
At the player level, Kaggle assigns each user a Code Point score—a proxy for coding skill—based on the popularity (measured by upvotes) of the user's publicly shared notebooks, with older notebooks receiving progressively less weight. Code Points are then mapped to discrete Code Tiers (i.e., Expert vs. Non-Expert) using Kaggle's proprietary ranking system.\footnote{Kaggle also uses secret algorithms to discount some votes to prevent abuse such as self-votes. \url{https://www.kaggle.com/progression/code}}
Detailed data statistics are reported in \cref{tb:summary_stats}. 

\begin{table}[htbp]
    \centering
    \captionsetup{font=small}
    \small
    \caption{\textbf{Summary Statistics}}
    \label{tb:summary_stats}
    \begin{tabularx}{0.85\textwidth}{lrrrrr} 
        \noalign{\hrule height 0.7pt}
        \hline \textbf{Variables} & \textbf{Count} & \textbf{Mean} & \textbf{Std} & \textbf{Min} & \textbf{Max} \\
        \hline
        \textbf{A. Competition Level}\\
        \textit{Number of players} & 28 & 1744.57 & 2172.56 & 28 & 8752 \\
        \textit{Reward quantity (USD)} & 28 & 26442.86 & 29941.03 & 0 & 100000 \\
        \textit{Number of prizes} & 28 & 2.71 & 1.65 & 1 & 8 \\
        \textit{Year} & 28 & 2014.75 & 3.35 & 2010 & 2021 \\
        \hline
        \textbf{B. Competition-Player Level}\\
        \textit{Final Performance} & 48848 & 0.80 & 0.14 & 0.01 & 1.00 \\
        \textit{Code Availability} & 48848 & 0.03 & 0.16 & 0 & 1 \\
        \hline
        \textbf{C. Player Level}\\
        \textit{Binary Type (Code Tiers)} & 36147 & 0.03 & 0.16 & 0 & 1 \\
        \textit{Continuous Type (Code Points)} & 36147 & 1.40 & 41.42 & 0 & 5700 \\
        \hline
        \noalign{\hrule height 0.7pt}
        \hline
    \end{tabularx}
    \vspace{0.1in}
    \captionsetup{font=footnotesize}
    \caption*{\textit{Note:} Panel~A is summarized across competitions: \textit{Number of players} counts the distinct players per competition. Panel~B is summarized across competition-player pairs: \textit{Code Availability} is 1 if the  associated code to the final submission is available. Panel~C is summarized across players: \textit{Binary type} equals 1 when the player is classified as an expert coder.}
\end{table}

A key data challenge is to measure creative vs mechanistic efforts as they are not observed directly from the data. We address this challenge by developing a novel method to quantify the extent of each type of effort from participants' code scripts.
We start by defining machine learning contest \emph{data pipeline} as the end-to-end workflow that turns raw competition data into a final submission. Such a data pipeline consists of three stages: data pre-processing and feature engineering, model training, and hyperparameter tuning and ensembling.
\cref{tab:pipeline} summarizes the purpose and practices of each stage. In data pre-processing and feature engineering, creative effort could be parsing passenger titles such as 'Doctor' out of raw name strings to proxy age and social status, whereas mechanistic effort could be taking the logarithm of a variable. In model training, creative effort could be upweighting small sellers in the loss after identifying that prediction errors are substantially larger for them, whereas mechanistic effort could be training a standard model using L2 loss. In hyperparameter tuning and ensembling, creative effort could be designing an ensemble that assigns different models to different regions of the input space after identifying that each model excels on different types of examples, whereas mechanistic effort could be averaging the predictions of multiple models with fixed weights. More examples are provided in Appendix~\ref{apx:technique_classification}.

\begin{table}[t]
	\centering
	\fontsize{10}{13}\selectfont
	\captionsetup{font=footnotesize}
	\caption{\bf Machine learning contest data pipeline.}
	\label{tab:pipeline}
	{\renewcommand{\arraystretch}{1.5}
	\begin{tabularx}{\textwidth}{p{0.2\textwidth}p{0.32\textwidth}X}
		\toprule
		Stage & Purpose & Practices \\
		\midrule
        Data pre-processing and feature engineering & Transform raw data into meaningful features (variables) by cleaning, selecting, manipulating, and creating input variables using domain knowledge to enable ML models to learn data patterns more effectively & Handle missing values, detect and treat outliers, remove duplicates, correct wrong data types, unify inconsistent categories, and fix formatting or labeling issues; then apply encoding for categorical variables, perform scaling or normalization, create interaction, aggregation, time-based, or text-based features, and use data augmentation to expand training examples.\\
		Model training (inner loop)  & Build models and improve their predictive performance on training data set & Choose model types, train models such as XGBoost, LightGBM, CatBoost, or neural networks. \\
        Hyperparameter tuning and ensembling (outer loop) & Tune hyperparameters to improve the performance of each model, and then combine multiple models to get a better overall prediction & Perform data splitting and cross-validation to select the best hyperparameters (e.g., learning rate, tree depth) for each model, then combine their outputs by averaging, voting or stacking.\\
		\bottomrule
	\end{tabularx}
    }
\end{table}

Having formalized the data pipeline, we design the following three-step procedure to construct measures of creative vs mechanistic effort for each contest-player pair:
\begin{enumerate}
	\item Partition each script into contiguous, non-overlapping logical blocks and assign each block to one of the three stages in the data pipeline defined in \cref{tab:pipeline} or none of them (see detailed prompt in Appendix~\ref{ax:prompt_pipeline});
	\item Identify all techniques used in each block and classify each of them as creative or mechanistic (see detailed prompt in Appendix~\ref{ax:prompt_technique});
	\item Count the number of creative and mechanistic techniques across all blocks to construct measures of creative vs mechanistic effort for each stage and for the pipeline as a whole.
\end{enumerate}

A logical block is a contiguous, non-overlapping sequence of code statements that jointly performs a single conceptual task within the machine learning pipeline. 
We first segment each script into logical blocks before assigning pipeline stages because code statements do not necessarily follow the sequential order of the data pipeline. For example, a participant may write code for feature engineering, then model training, and later realizes they need to reengineer features. This accommodates the iterative nature of code development and prevents missing-out and misclassification of out-of-order code. Boilderplate  — imports, config, plotting, prints, blanks — is left unblocked.

Each technique is an operation — a single function call, model fit, evaluation step, or contiguous statement that produces or transforms a modeling artifact. For example, the following code snippet is a data preprocessing and feature engineering block and each line is a technique that engineers a feature from the raw data. Likewise, we ignore boilerplate lines.
Whenever it is tricky to distinguish creative (especially creative problem-specific
adaptation) from mechanistic techniques, we conduct a thought experiment by asking ``would any data scientist hand this same dataset, without knowing anything about the competition, write this exact code as a default first move?" If yes, we classify it as mechanistic and vice versa.
\begin{lstlisting}[title={\bf Example of logical block and technique}, label={lst:feature_engineering}, basicstyle=\ttfamily\small, columns=fullflexible, keepspaces=true, backgroundcolor=\color{gray!12}]
# data preprocessing and feature engineering block
1. df["age"] = np.log(df["age"]) 		# technique 1: mechanistic
2. df["income"] = df["income"].fillna(0) 	# technique 2: mechanistic
\end{lstlisting}

We use gpt-4.1-mini to perform this classification task. For each technique initially classified as creative, we validate the classification using a more advanced model, gpt-4.1, and find that 63.2\% of the classifications are consistent. We use the more advanced model's judgment to resolve any disagreements. \cref{tb:effort_stats} reports the summary statistics of creative vs mechanistic techniques across all contest-player pairs whose code is accessible. 
On average, players employ substantially more mechanistic techniques (13.19) than creative techniques (0.94). Expert players employ twice as many creative techniques as non-expert players (1.87 vs 0.85), but only slightly more mechanistic techniques (15.95 vs 12.92). Across pipeline stages, the majority of both creative (76\%) and mechanistic (54\%) techniques are concentrated in data preprocessing and feature engineering.

\begin{table}[htbp]
    \centering
    \captionsetup{font=small}
    \small
    \caption{\textbf{Technique Composition by Player}}
    \label{tb:effort_stats}
    \begin{tabularx}{0.9\textwidth}{lrrrrr} 
        \noalign{\hrule height 0.7pt}
        \hline \textbf{Variables} & \textbf{Count} & \textbf{Mean} & \textbf{Std} & \textbf{Min} & \textbf{Max} \\
        \hline
        \textit{All-stages creative techniques} & 1355 & 0.94 & 2.07 & 0 & 36 \\
        \textit{All-stages mechanistic techniques} & 1355 & 13.19 & 8.52 & 0 & 74 \\
        \textit{All-stages creative techniques (High type)} & 126 & 1.87 & 2.77 & 0 & 17 \\
        \textit{All-stages creative techniques (Low type)} & 1182 & 0.85 & 1.98 & 0 & 36 \\
        \textit{All-stages mechanistic techniques (High type)} & 126 & 15.95 & 11.28 & 0 & 59 \\
        \textit{All-stages mechanistic techniques (Low type)} & 1182 & 12.92 & 8.10 & 1 & 74 \\
        \textit{Stage-1 creative techniques} & 1355 & 0.71 & 1.73 & 0 & 32 \\
        \textit{Stage-2 creative techniques} & 1355 & 0.15 & 0.56 & 0 & 10 \\
        \textit{Stage-3 creative techniques} & 1355 & 0.08 & 0.34 & 0 & 3 \\
        \textit{Stage-1 mechanistic techniques} & 1355 & 7.13 & 6.46 & 0 & 58 \\
        \textit{Stage-2 mechanistic techniques} & 1355 & 3.06 & 2.47 & 0 & 29 \\
        \textit{Stage-3 mechanistic techniques} & 1355 & 2.99 & 2.76 & 0 & 34 \\
        \noalign{\hrule height 0.7pt}
        \hline
    \end{tabularx}
    \vspace{0.1in}
    \captionsetup{font=footnotesize}
    \caption*{\textit{Note:} Summarized across contest-player pairs whose submission code was accessible. \textit{All-stages creative techniques} and \textit{All-stages mechanistic techniques} count creative and mechanistic techniques summed over all pipeline stages; the per-stage rows give the corresponding counts within stages 1--3. Rows marked Low/High type split all-stages counts by binary type.}
\end{table}

\paragraph{Empirical Evidence of Model Predictions.}
We generate a series of empirical facts from the data, from which we observe consistency with our model predictions as elaborated below. 

\vspace{0.1in}\noindent
{\textbf{Test for Proposition \ref{prop:exist}.}} We first test the monotonicity of the equilibrium, in particular, whether model performance is monotone in type. 
\Cref{fig:type_score} is a scatter plot for model performance of each player type at competition level and highlights three patterns. First, almost all players achieve an AUC score of at least 0.5—the expected performance under random guessing—consistent with the standard interpretation of AUC in binary classification. Second,
the lower end of these scores concentrates around 0.5 across competitions, suggesting that a nontrivial portion of players perform close to random guessing. In contrast, the upper end varies across competitions, typically clustering around 0.8 and 0.9, reflecting substantial heterogeneity in competition difficulty. Third, high types achieve higher scores on average, providing visual evidence consistent with the monotonicity prediction of the equilibrium.

\begin{figure}
	\caption{\bf Performance distribution by types}
	\centering
	\includegraphics[width=1\textwidth]{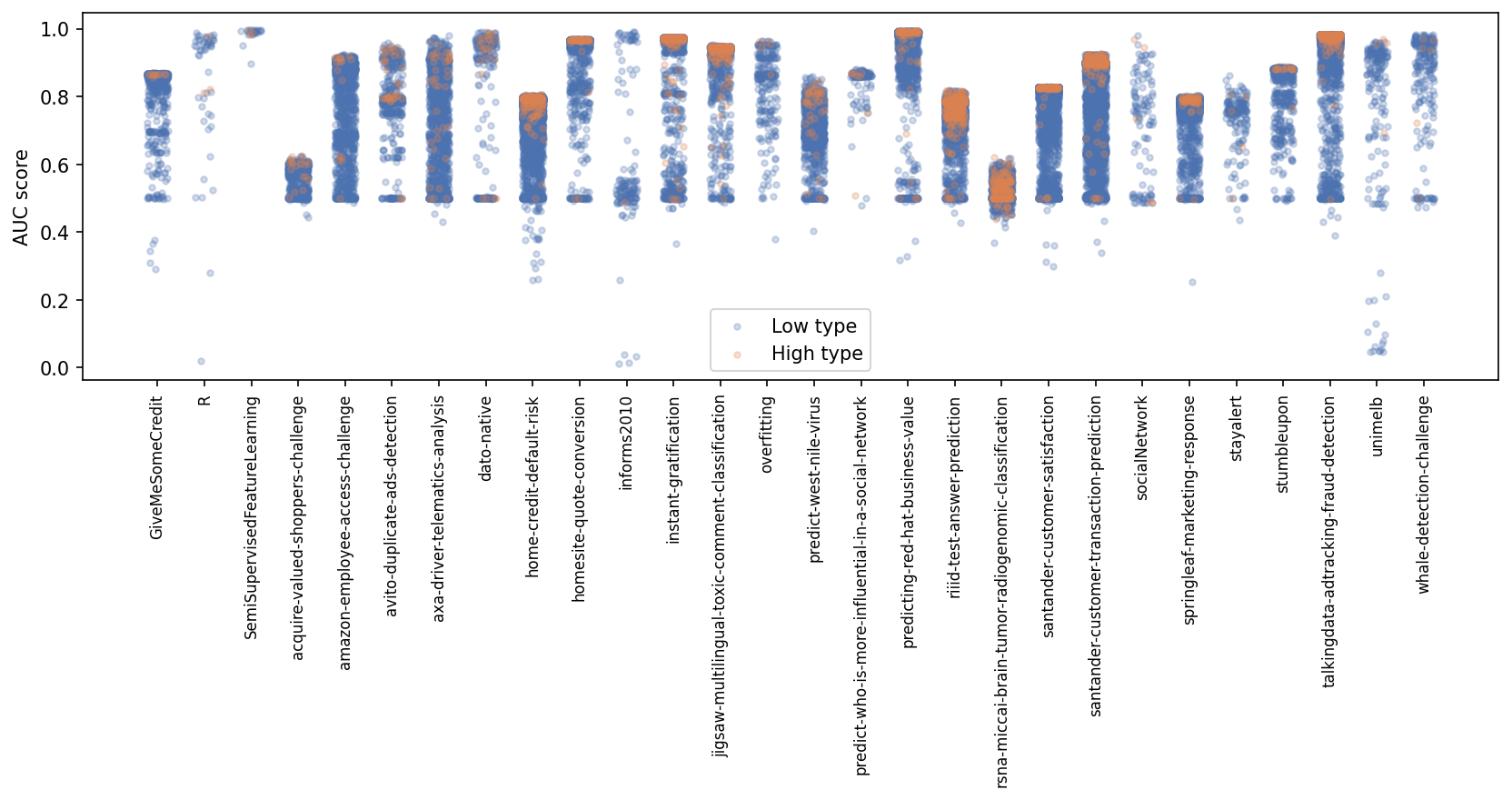}
    \captionsetup{font=scriptsize}
	\label{fig:type_score}
\end{figure}

Formally, we estimate the impact of type on model performance using
\begin{equation*}
    \signal_{ij}=\eta_0+\eta_1\cdot\type_i+\tau_j+\varepsilon_{ij},
\end{equation*}
where $\signal_{ij}$ is the model performance of player $i$ in contest $j$, $\type_i$ is the binary type of player $i$,
$\tau_j$ is contest fixed effect and is included to account for variations in the performance specific to the distinct nature of different prediction tasks, $\varepsilon_{ij}$ captures the remaining noises in the performance, such as that induced by the random splitting of the training data for cross-validation. 
The result is shown in column (1) of Table~\ref{tab:regressions}: compared to non-expert players, expert players' model performance is on average 3.6\% higher, \emph{ceteris paribus}, suggesting that the equilibrium model performance is monotone in player type, consistent with our finding in Proposition \ref{prop:exist}.

\vspace{0.2in}\noindent\textbf{Test for Proposition \ref{prop:compare}.} Proposition \ref{prop:compare} implies that higher-type players are more likely to exert creative effort. It also suggests that sufficiently low-type players exert positive, type-independent mechanistic effort. Moreover, if equilibrium creative effort increases sufficiently rapidly with type, equilibrium mechanistic effort decreases with type because high-type players optimally reallocate effort from mechanistic to creative activities, whose marginal returns increase with type. We test these predictions by estimating
\begin{equation}
    y_{ij}=\eta_0+\eta_1\cdot\type_i+\tau_j+\varepsilon_{ij},\quad y_{ij}\in\{\mdeffort_{ij},\pteffort_{ij}\},
\end{equation}
where creative and mechanistic effort, $\mdeffort_{ij}$ and $\pteffort_{ij}$, are proxied by the number of creative and mechanistic techniques employed, respectively. All other variables are defined as before. Column (2) of \cref{tab:regressions} suggests that, \emph{ceteris paribus}, expert players employ 0.41 more creative techniques than non-expert players, providing empirical support for the theoretical prediction that higher-type players are more likely to exert creative effort. Column (3) shows that expert and non-expert players use a statistically indistinguishable number of mechanistic techniques. The lack of statistical significance is likely driven by the fact that more than 95\% of players are non-experts. Consequently, the estimate primarily reflects the type-independent mechanistic effort predicted among low-type players.
Nevertheless, its negative sign is consistent with the prediction that mechanistic effort declines with type.

\vspace{0.2in}\noindent\textbf{Validation of the Specification in \cref{eq:fitness}.} A key assumption in this paper is that, as is demostrated in \cref{eq:fitness}, the return to creative effort is type-dependent while the return to mechanistic effort is not. Although we have provided descriptive rationales for this assumption, we cannot completely rule out the possibility that the returns to mechanistic effort also vary with player type. For example, higher-type players may be familiar with a wider range of off-the-shelf techniques and therefore earn higher returns from mechanistic effort. 
To provide empirical support for this assumption,
we run the following regression:
\begin{equation}\label{reg:specification}
    \signal_{ij}=\eta_0+\eta_1\cdot\mdeffort_{ij}+\eta_2\cdot\pteffort_{ij}+\eta_3\cdot(\type_i\times\mdeffort_{ij})+\eta_4\cdot(\type_i\times\pteffort_{ij})+\tau_j+\varepsilon_{ij},
\end{equation}
where the two interaction terms  capture the (in)dependence of the returns to each effort on player type. 
The results are presented in column (4) of Table \ref{tab:regressions}: the coefficient on the interaction between creation and type is positive (0.441) and significant, suggesting that an additional unit of creative effort yields a 0.441\% higher performance for high-type players than for low-type players. In contrast, the interaction for mechanization is statistically insignificant, suggesting that its returns are type-independent. These findings are consistent with our model specification in \cref{eq:fitness}.

One limitation of this approach is that our proxy for player type may not fully capture heterogeneity in underlying ability. In particular, players may differ along additional dimensions, such as domain knowledge, diagnostic problem-solving, computing capacity, that are not reflected in our type measure. 
In software engineering, this may include application-specific domain knowledge; in legal practice, communication fluency and verbal expressiveness; and in machine learning, access to high-performance GPUs or TPUs.
If these unobserved characteristics affect both the model performance and the effort choices, the estimates could be biased. Therefore, our estimates should be interpreted as evidence consistent with 
type-dependent returns to creative effort and type-independent returns to mechanistic effort, rather than as estimates of the underlying structural return functions.
In addition, the regression results cannot rule out underlying complementarity between mechanistic effort and player type. Even if the two are complementary, their observed relationship may still be insignificant—or even negative—if this complementarity is sufficiently weak relative to that between creative effort and type, leading high-type players to shift effort toward the creative margin.
Nevertheless, the results indicate stronger complementarity between creative effort and type than between mechanistic effort and type, a pattern captured by our model specification in \cref{eq:fitness}.

\begin{table}
    \centering
    \caption{\bf Regression results}
    \label{tab:regressions}
    \hspace*{-0cm}
    \scalebox{0.85}{%
        {
\def\sym#1{\ifmmode^{#1}\else\(^{#1}\)\fi}
\begin{tabular}{l*{5}{D{.}{.}{-1}}}
\toprule
            &\multicolumn{1}{c}{(1)}&\multicolumn{1}{c}{(2)}&\multicolumn{1}{c}{(3)}&\multicolumn{1}{c}{(4)}\\
            &\multicolumn{1}{c}{Performance (\%)}&\multicolumn{1}{c}{Creative}&\multicolumn{1}{c}{Mechanistic}&\multicolumn{1}{c}{Performance (\%)}\\
\midrule
High type   &       3.576\sym{***}&       0.406\sym{**} &      -0.065         &                     \\
            &     (0.234)         &     (0.167)         &     (0.722)         &                     \\
\addlinespace[1.1em]
Creative    &                     &                     &                     &       0.855\sym{***}\\
            &                     &                     &                     &     (0.119)         \\
\addlinespace[1.1em]
Mechanistic &                     &                     &                     &       0.049\sym{*}  \\
            &                     &                     &                     &     (0.026)         \\
\addlinespace[1.1em]
Creative $\times$ High Type&                     &                     &                     &       0.441\sym{*}  \\
            &                     &                     &                     &     (0.233)         \\
\addlinespace[1.1em]
Mechanistic $\times$ High Type&                     &                     &                     &       0.029         \\
            &                     &                     &                     &     (0.039)         \\
\addlinespace[1.1em]
Constant    &      79.788\sym{***}&       0.910\sym{***}&      13.222\sym{***}&      79.562\sym{***}\\
            &     (0.050)         &     (0.049)         &     (0.212)         &     (0.369)         \\
\midrule
Observations&\multicolumn{1}{c}{46766}         &\multicolumn{1}{c}{1308}         &\multicolumn{1}{c}{1308}         &\multicolumn{1}{c}{1308}         \\
R\textsuperscript{2}&\multicolumn{1}{c}{0.440}         &\multicolumn{1}{c}{0.366}         &\multicolumn{1}{c}{0.284}         &\multicolumn{1}{c}{0.710}         \\
Competition FE&\multicolumn{1}{c}{Yes}         &\multicolumn{1}{c}{Yes}         &\multicolumn{1}{c}{Yes}         &\multicolumn{1}{c}{Yes}         \\
\bottomrule
\multicolumn{5}{l}{\footnotesize Standard errors in parentheses}\\
\multicolumn{5}{l}{\footnotesize \sym{*} \(p<0.1\), \sym{**} \(p<0.05\), \sym{***} \(p<0.01\)}\\
\end{tabular}
}

    }
\end{table}


 
\section{Conclusion}

This paper studies how the diffusion of generative AI reshapes performance-based
screening when agents can allocate effort between creative and mechanistic tasks.
We make three contributions. First, we develop a game-theoretic contest model in
which a principal screens agents on observable performance, but agents differ in
private ability and choose how to divide effort between \emph{creation}, whose
return is complementary to ability, and \emph{mechanization}, which AI renders
cheap and largely ability-independent. Second, within this framework we formally
define \emph{performance manipulation} as the strategic over-investment in
AI-augmented mechanistic effort induced by competitive pressure, and
characterize the conditions under which it arises in equilibrium. Third, we test
the model using a novel, language-model-based methodology that measures
multidimensional effort from nearly 1{,}500 publicly shared code scripts in
Kaggle competitions.

Our central result is a threshold in AI capability. When AI is sufficiently
advanced, the marginal return to mechanistic effort becomes so high that it
crowds out creative effort for agents of every type; the cost of performance
ceases to depend on ability, the single-crossing property fails, and the market
collapses into a pooling equilibrium. When AI capability is more modest, the
complementarity between creative effort and ability preserves single-crossing and
sustains a separating equilibrium. In this regime it is the low-ability agents
who manipulate: they flood the evaluation with AI-augmented mechanistic output to
inflate their observable scores; while high-ability agents distinguish
themselves through creation. Our Kaggle evidence is consistent with this
characterization: high-type participants exert almost 50\% more creative effort
than low types but no more mechanistic effort, and only the interaction of creative
effort with ability raises final performance.

These results carry a constructive message for labor markets under AI. The recent
retreat from asynchronous evaluations is well-founded for tasks dominated by
routine implementation, but screening does not collapse wholesale. Instead, it can be
restored by design. Principals should redesign evaluations to embed a dominant
creative component, favoring tasks that demand structural synthesis, diagnostic
problem-solving, or judgment under uncertainty over those solvable by routine
execution, and should adopt more sharply skewed reward structures that
compensate scarce creative effort and thereby restrain manipulation. We view this
work as a step toward a more concrete characterization of optimal ML contest and
evaluation design that can better address misaligned incentives and produce more
desirable outcomes for organizers. Natural extensions include endogenizing the
principal's choice of task composition, allowing AI capability and contestant
ability to co-evolve, and studying dynamic contests in which manipulation and
detection interact over time.

\newpage
\bibliographystyle{plainnat}
\bibliography{reference}

\newpage

\appendix
\section{Illustrative Examples Across Four Applications}
\label{apx:effort_examples_four_applications}

\begin{table}[htbp]
    \centering
    \fontsize{10}{13}\selectfont
    \captionsetup{font=footnotesize}
    \caption{\bf Creative and Mechanistic Effort Across Four Applications}
    \label{tab:effort_examples_four_applications}
    {\renewcommand{\arraystretch}{1.4}
    \begin{tabularx}{\textwidth}{p{0.24\textwidth}X X}
        \toprule
        Application & Mechanistic effort & Creative effort \\
        \midrule
        Machine learning &
        \begin{itemize}[leftmargin=*, nosep]
            \item Deploy standard plug-and-play packages.
            \item Conduct brute-force hyperparameter searches.
            \item Use default architectures.
        \end{itemize}
        &
        \begin{itemize}[leftmargin=*, nosep]
            \item Introduce novel algorithmic frameworks.
            \item Design new optimization strategies.
            \item Feature engineering based on domain knowledge.
        \end{itemize}
        \\
        Software engineering &
        \begin{itemize}[leftmargin=*, nosep]
            \item Write boilerplate code.
            \item Translate logic across languages.
            \item Generate routine unit tests.
        \end{itemize}
        &
        \begin{itemize}[leftmargin=*, nosep]
            \item Design unprecedented system architectures.
            \item Navigate ambiguous client requirements.
            \item Anticipate novel hardware bottlenecks.
        \end{itemize}
        \\
        Legal practice &
        \begin{itemize}[leftmargin=*, nosep]
            \item Conduct e-discovery.
            \item Draft boilerplate contracts (e.g., NDAs).
            \item Verify legal citations.
        \end{itemize}
        &
        \begin{itemize}[leftmargin=*, nosep]
            \item Formulate novel litigation strategies.
            \item Assess opponent psychology.
            \item Structure first-of-its-kind M\&A deals in emerging regulatory environments.
        \end{itemize}
        \\
        Scientific research &
        \begin{itemize}[leftmargin=*, nosep]
            \item Survey existing literature.
            \item Execute pre-specified statistical regressions.
            \item Format manuscripts.
        \end{itemize}
        &
        \begin{itemize}[leftmargin=*, nosep]
            \item Generate novel hypotheses.
            \item Design new identification strategies.
            \item Use intuition to process negative results or data anomalies.
        \end{itemize}
        \\
        \bottomrule
    \end{tabularx}
    }
\end{table}
\section{Technique Classification}
\label{apx:technique_classification}

\subsection{Examples.}\label{apx:technique_classification_examples}

We use the Kaggle competition \emph{Predicting Heart Disease} as an example to illustrate the distinction between creative and mechanistic effort. Contestants are given a tabular dataset including features of age, sex, chest pain type, blood pressure etc., and a label indicating whether a patient has heart disease, with the goal of predicting the likelihood of heart disease.\footnote{\url{https://www.kaggle.com/competitions/playground-series-s6e2/overview}} 
We map the solution of the first-place team, \emph{Masaya Kawamata}, to each of the three stages of the data pipeline and classify the techniques used at each stage as either creative or mechanistic. The results are presented in Tables~\ref{tab:effort_classification_featureengineering}, \ref{tab:effort_classification_training_evaluation}, and \ref{tab:effort_classification_stage3}.
\begin{table}[htbp]
	\centering
	\fontsize{10}{13}\selectfont
	\captionsetup{font=footnotesize}
	\caption{\bf Technique Classification in Data Pre-processing and Feature Engineering}
	\label{tab:effort_classification_featureengineering}
	{\renewcommand{\arraystretch}{1.5}
	\begin{tabularx}{\textwidth}{p{0.34\textwidth}p{0.14\textwidth}X}
		\toprule
		Technique & Classification & Reasoning \\
		\midrule
		Quantile-based binning, equal-width binning, simple rounding & mechanistic & Standard feature preprocessing techniques (like log transformation and data rescaling) applied without domain-specific adaptation. \\
		Digit feature extraction (units, tens, decimal digits) & creative & Creative problem-specific adaptation --- recognizing that digit positions might capture hidden structure in the data is an innovative, non-standard idea. \\
		Treating all features as categorical (string conversion) & creative & Creative problem-specific adaptation --- reframing numerical features as categorical to change how tree models handle them is a non-obvious representation choice. \\
		Frequency encoding & mechanistic & Standard, plug-and-play preprocessing technique commonly applied without adaptation. \\
		Genetic programming features (gplearn) & mechanistic & Experimenting with default feature generation tools without adapting them to the domain or data structure --- automated trial-and-error. \\
		Denoising Variational Autoencoder (DVAE) for latent representations & creative & Foundational innovation --- designing and training a DVAE to produce alternative nonlinear representations goes beyond standard tools and requires proposing a new component in the pipeline. \\
		Creating multiple feature sets for diversity (not searching for one optimal set) & creative & Creative problem reformulation --- deliberately generating diverse representations rather than optimizing one feature set reframes the feature engineering problem to capture the underlying structure of ensembling. \\
		\bottomrule
	\end{tabularx}
	}
\end{table}

\begin{table}[htbp]
	\centering
	\fontsize{10}{13}\selectfont
	\captionsetup{font=footnotesize}
	\caption{\bf Technique Classification in Model Training}
	\label{tab:effort_classification_training_evaluation}
	{\renewcommand{\arraystretch}{1.5}
	\begin{tabularx}{\textwidth}{p{0.34\textwidth}p{0.14\textwidth}X}
		\toprule
		Technique & Classification & Reasoning \\
		\midrule
		Training XGBoost, LightGBM, CatBoost & Mechanistic & Experimenting with default model architectures --- standard plug-and-play choices for tabular competitions. \\
		Training TabICL & Creative & Creative problem-specific adaptation --- choosing an uncommon in-context learning model for its diversity contribution, not its standalone performance, reflects innovative reasoning about ensemble composition. \\
		Training AutoGluon & Mechanistic & Standard plug-and-play automated ML tool. \\
		Using 5-fold StratifiedKFold with fixed seed & Mechanistic & Standard practice --- follows convention with no innovation. \\
		Standard hyperparameter tuning for each model & Mechanistic & Brute-force searching the hyperparameter space to optimize performance --- explicitly mechanistic by definition. \\
		Retraining on full data with 20 random seeds and averaged predictions & Mechanistic & Standard trial-and-error optimization --- seed averaging is a well-known mechanical technique. \\
		Setting n\_estimators to 1.25$\times$ average best iteration for full retraining & Creative & Creative problem-specific adaptation --- the specific multiplier was derived from experimental insight about how full-data retraining behaves, not from a standard recipe. \\
		Choosing Ridge regression as meta-learner & Creative & Creative problem-specific adaptation --- after testing flexible meta-models and finding they overfit, the contestant made a principled design decision that simplicity better reflects the nature of the ensembling problem. \\
		Overall strategy: diversity $\rightarrow$ selection $\rightarrow$ simple combination & Creative & Proposing a new algorithmic framework --- designing the entire system around controlled diversity rather than single-model dominance is a foundational innovation that reformulates the problem. \\
		\bottomrule
	\end{tabularx}
	}
\end{table}

\begingroup
\fontsize{10}{13}\selectfont
\captionsetup{font=footnotesize}
\renewcommand{\arraystretch}{1.5}
\begin{xltabular}{\textwidth}{p{0.34\textwidth}p{0.14\textwidth}X}
		\caption{\bf Technique Classification in Hyperparameter Tuning and Ensembling}
		\label{tab:effort_classification_stage3}\\
		\toprule
		Technique & Classification & Reasoning \\
		\midrule
		\endfirsthead
		\multicolumn{3}{l}{\tablename~\thetable\ -- continued from previous page}\\
		\toprule
		Technique & Classification & Reasoning \\
		\midrule
		\endhead
		\midrule
		\multicolumn{3}{r}{\itshape continued on next page}\\
		\endfoot
		\bottomrule
		\endlastfoot

		Using 5-fold Stratified Cross-Validation for model evaluation and tuning &
		Mechanistic &
		Standard validation procedure widely used in tabular ML competitions. The fold structure follows established practice and does not encode competition-specific insight. \\
		
		Hyperparameter tuning of individual XGBoost, LightGBM, CatBoost, TabICL, AutoGluon and neural-network models &
		Mechanistic &
		Searching learning rates, tree depths, regularization settings, and related parameters is routine trial-and-error optimization. This is explicitly an off-the-shelf procedure. \\
		
		Generating large numbers of model variants through different hyperparameter settings &
		Mechanistic &
		Creating multiple model variants by perturbing hyperparameters is a common ensemble-building recipe and does not itself incorporate domain knowledge about the heart-disease dataset. \\
		
		Evaluating CV--LB (cross-validation versus leaderboard) consistency before selecting models &
		Creative &
		Problem-specific adaptation. Rather than maximizing CV alone, the author explicitly analyzed the relationship between validation and leaderboard behavior and selected models that generalized consistently. This reflects reasoning about competition dynamics rather than a default workflow. \\
		
		Prioritizing model diversity during ensemble construction &
		Creative &
		Problem-specific adaptation. The ensemble was intentionally designed to contain models with different error patterns rather than simply the strongest individual models. This requires judgment about complementarity among models and is not a default first step. \\
		
		Selecting ensemble members based on marginal contribution rather than standalone score &
		Creative &
		Creative adaptation. Models were retained because they improved the ensemble as a whole, even when their individual performance was weaker. This reflects non-obvious reasoning about interaction effects among predictors. \\
		
		Diversity-selection framework for choosing candidate models &
		Creative &
		Creative problem reformulation. The author organized the search process around generating diverse candidates and then filtering them according to ensemble utility, rather than pursuing a single best model. This changes how the optimization problem is approached. \\
		
		Simple weighted averaging/blending of selected models &
		Mechanistic &
		Averaging predictions from multiple models is a standard ensemble technique routinely used in Kaggle competitions and production ML systems. \\
		
		Stacking models using out-of-fold predictions &
		Mechanistic &
		Stacking is a well-established ensemble recipe. Although effective, using OOF predictions as inputs to a second-level model is standard practice and does not constitute a novel adaptation by itself. \\
		
		Choosing Ridge regression as the stacking meta-learner &
		Creative &
		Problem-specific adaptation. After experimenting with more flexible meta-models and observing overfitting, the author deliberately chose a simpler linear combiner whose inductive bias better matched the ensemble-selection problem. \\
		
		Retraining selected models on the full training data after tuning &
		Mechanistic &
		Standard post-tuning procedure. Once hyperparameters are selected, retraining on all available data is conventional practice. \\
		
		Averaging predictions from multiple random seeds (approximately 20 seeds per model) &
		Mechanistic &
		Seed averaging is a widely known variance-reduction technique and represents routine optimization rather than a novel adaptation. \\
		
		Using approximately $1.25\times$ the average best iteration when retraining on the full dataset &
		Creative &
		Problem-specific adaptation. The multiplier was derived empirically from observed behavior during full-data retraining and is not a standard rule found in typical workflows. \\
		
		Overall strategy: diversity $\rightarrow$ selection $\rightarrow$ simple combination &
		Creative &
		Creative system-level design. The solution's central contribution was treating ensemble construction as a diversity-management problem rather than a search for a single dominant model. This organizing principle guided model generation, filtering, and final blending. \\

\end{xltabular}
\endgroup

\subsection{Prompts}\label{ax:prompts}

\subsubsection{Data pipeline classification.}\label{ax:prompt_pipeline}
\begin{promptbox}
You are an expert at reading ML competition code. Your job is to segment the given script into contiguous, non-overlapping logical blocks and assign each block exactly one of the three ML-pipeline stages.

Rules:

* Segment by what the code is doing, not by cell or function boundaries; merge adjacent lines that serve the same stage into one block.
* Cover as many lines as you can with stage 1-3 labels. Pure boilerplate (imports, library config, plotting/EDA that does not transform features, print statements, comments-only ranges, blank lines) should NOT be classified -- simply leave them out of the block list; do not invent a 4th stage.
* Function / class definitions: if the body performs one of stages 1-3, classify the ENTIRE def or class block (including the signature line and any decorators) by what the body does. Only skip definitions that are pure utilities not fitting any stage (e.g., math helpers like `rmse`, logging wrappers, IO wrappers, generic decorators). If a function body spans multiple stages, split the def block into contiguous sub-blocks at the stage boundaries inside the body.
* Blocks must be contiguous (start line <= end line), non-overlapping, and in order.
* If a stage truly does not appear in the script, simply omit it.
\end{promptbox}

\subsubsection{Technique classification.}\label{ax:prompt_technique}
\begin{promptbox}
You are an expert at reading ML competition code. You will be shown a SLICE
of a script that has already been labeled as belonging to a particular
ML-pipeline stage. Your job is to:

  1. Enumerate the distinct ML techniques used in this slice. ONE TECHNIQUE
     = ONE OPERATION: a single function call, model fit, evaluation step, or
     contiguous statement that produces or transforms a modeling artifact.
     Components inside ONE call (e.g., a `gam(...)' call whose formula uses
     both `s(dWeek)' and `lo(Latitude, Longitude)' as smoothers) belong to
     the SAME technique -- do NOT split them. A different statement on a
     different line that does something else (e.g., computing AUC, refitting
     on full data, writing the submission CSV) is a SEPARATE technique. Do
     NOT list pure boilerplate (imports, library config, IO of raw data,
     prints, plotting).

  2. For EACH technique, classify the effort as exactly one of {`creative',
     `mechanistic'} using their definitions, judging the operation AS A
     WHOLE -- including its arguments and component pieces. Example: a
     `gam(...)' call whose formula contains a `lo(Latitude, Longitude)'
     spatial smoother counts as CREATIVE because of that problem-specific
     spatial adaptation, even though `gam' itself is a standard library
     function.

Rules:
  * If the slice contains no genuine technique (only boilerplate/utilities),
    return an empty list.
  * Each technique should be a short noun phrase (<= 10 words).
  * Technique line ranges MUST be within this slice's line ranges and NON-OVERLAPPING. For each technique, start line <= end line, and both start and end lines must be on real code (not blank or comment-only). 
  * Each code line may appear in AT MOST ONE technique (a line may appear in zero
    techniques if it is boilerplate / not part of any technique). If two
    operations naturally share a line (e.g., a `for' loop header wrapping
    both feature construction and model fitting on the same line), pick whichever
    the line's primary purpose serves.
\end{promptbox}

\section{Equilibrium Existence}\label{apx:existence}

\subsection{The Auxiliary Game} \label{apx:auxiliary_game}
An equilibrium of this auxiliary game is a strategy profile $((\eqmfitnessalt_i^\ddag)_i, \eqmeffortalts^\ddag)$ that satisfies two conditions. First, given the strategy profile of other players $\eqmfitnessesalt_{-i}^\ddag$, for each player $i$, $\eqmfitnessalt_{i}^\ddag$ solves the following maximization problem:
\begin{equation}\tag{$\mathcal{P}_i'$}\label{alternative problem i}
    \begin{aligned}
      &  \max_{\fitness_i}  \quad \mathbb{E}_{\typesminus, \signals}\left[\sum_{k=1}^\numplayer \indicate{i=i_k^*}\cdot \reward_k\bigg| \fitness_i,\eqmfitnessesalt^\ddag_{-i}\right]+\fitness_i-\fscostfn(\fitness,\type;\gamma) \\
        =&  \max_{\fitness_i}  \quad \contestgain(\fitness_i,\eqmfitnessesalt^\ddag_{-i})+\fitness_i-\fscostfn(\fitness,\type;\gamma),
    \end{aligned}
\end{equation}
where the expected contest gain is given by:
\begin{equation*}
    \begin{aligned}
   \contestgain(\fitness_i,\eqmfitnessesalt^\ddag_{-i}) &= \mathbb{E}_{\typesminus, \signals}\left[\sum_{k=1}^\numplayer \indicate{i=i_k^*}\cdot \reward_k\bigg|\fitness_i, \eqmfitnessesalt^\ddag_{-i}\right] \\
    &= \sum_{k=1}^\numplayer \prob{i=i_k^*|\fitness_i,\eqmfitnessesalt^\ddag_{-i}}\cdot \reward_k
    \end{aligned}
\end{equation*} 
Second, $\eqmeffortalts^\ddag(\fitness,\type)$ represents the optimal effort allocation given a targeted mean performance $\fitness$ and type $\type$:  
\begin{equation}\label{second step}
    \eqmeffortalts^\ddag(\fitness,\type)=\argmin_{\mdeffort,\pteffort\geq 0}\etcostfn(\mdeffort+\pteffort) \quad \text{s.t.} \quad  \mdeffortfn(\mdeffort,\type)+\pteffortfn(\pteffort;\gamma)=\fitness.
\end{equation}

\begin{proof}[Proof of Lemma \ref{lemma:equivalent}]
Consider any player $i$ in the original game. For any player $j\neq i$, player $j$'s strategy affects player $i$'s payoff exclusively by altering player $i$'s ranking in the contest. This influence is entirely summarized by player $j$'s mean performance. Fixing the mean performance profile of other players $\fitnesses_{-i}$, player $i$'s optimal choice of mean performance is identical in both the original game and the auxiliary game. Because problem (\ref{player i's problem}) and problem (\ref{alternative problem i}) are fundamentally the same maximization problem when $\fitnesses_{-i}$ is fixed, they admit the same solution.
\end{proof}

\subsection{Existence of Symmetric Monotone Equilibrium}\label{apx:existence_proof}
We now prove that there exists at least one symmetric monotone equilibrium in the auxiliary game, which extends to the original game via Lemma \ref{lemma:equivalent}. We begin by characterizing the properties of the contest gain function $\contestgain(\fitness_i,\eqmfitnessesalt_{-i})$. 

\begin{lemma}\label{lemma:nondecreasing_contestgain}
    Under Assumption \ref{a:fosd}, for any $\eqmfitnessesalt_{-i}$, the contest gain function $\contestgain(\fitness_i,\eqmfitnessesalt_{-i})$ is non-negative, bounded, continuous, and non-decreasing in $\fitness_i$.
\end{lemma}

\begin{proof}[Proof of \cref{lemma:nondecreasing_contestgain}]
Non-negativity and boundedness follow straightforwardly from the problem definition. Continuity follows from the continuity of the payoff function and the continuous distribution function $\distsignal_{\fitness_i}$.

To prove the function is non-decreasing in $\fitness_i$, let $\pmf_{\fitness_i,\eqmfitnessesalt_{-i}}(k)=\prob{i=i_k|\fitness_i,\eqmfitnessesalt_{-i}}$ denote the probability that player $i$ achieves exactly rank $k$. The expected contest gain is:
\begin{equation*}
    \contestgain(\fitness_i,\eqmfitnessesalt_{-i}) = \sum_{k=1}^\numplayer \pmf_{\fitness_i,\eqmfitnessesalt_{-i}}(k)\cdot\reward_k
\end{equation*}
Let $\cmf_{\fitness_i,\eqmfitnessesalt_{-i}}(k)=\sum_{l=1}^k \pmf_{\fitness_i,\eqmfitnessesalt_{-i}}(l)$ denote the cumulative probability that player $i$ secures a rank of $k$ \emph{or better}. Applying summation by parts, we can rewrite the expected contest gain as:
\begin{equation*}
    \contestgain(\fitness_i,\eqmfitnessesalt_{-i}) = \sum_{k=1}^{\numplayer-1} \cmf_{\fitness_i,\eqmfitnessesalt_{-i}}(k) \cdot (\reward_k - \reward_{k+1}) + \reward_\numplayer \underbrace{\cmf_{\fitness_i,\eqmfitnessesalt_{-i}}(\numplayer)}_{=1}
\end{equation*}

\begin{lemma}\label{lemma:winprob_fosd}
    Under Assumption \ref{a:fosd}, for any $\fitness_i>\fitness_i'$ and any $\eqmfitnessesalt_{-i}$, $ \cmf_{\fitness_i,\eqmfitnessesalt_{-i}}(k)\geq  \cmf_{\fitness_i',\eqmfitnessesalt_{-i}}(k)$.
\end{lemma}

By \cref{lemma:winprob_fosd}, for any $\fitness_i > \fitness'_i$, the cumulative probability of achieving any rank $k$ or better strictly increases: $\cmf_{\fitness_i,\eqmfitnessesalt_{-i}}(k) \geq \cmf_{\fitness'_i,\eqmfitnessesalt_{-i}}(k)$. 
Since the prize vector is weakly ordered ($\reward_k - \reward_{k+1} \geq 0$), every term in the summation is weakly greater under $\fitness_i$. Thus, $\contestgain(\fitness_i,\eqmfitnessesalt_{-i}) \geq \contestgain(\fitness'_i,\eqmfitnessesalt_{-i})$, proving the expected gain is non-decreasing in mean performance.
\end{proof}

\begin{proof}[Proof of \cref{lemma:winprob_fosd}]
Let $\signals_{-i} = (\signal_j)_{j \neq i}$ denote the random performance profile of all other players, drawn independently from their respective distributions $\distsignal_{\fitness_j}$ given the strategy profile $\eqmfitnessesalt_{-i}$.

For any fixed realization of player $i$'s performance $\signal_i = s$, let $G_k(s \mid \eqmfitnessesalt_{-i})$ denote the probability that $s$ achieves a rank of $k$ or better. Achieving rank $k$ or better is equivalent to the event that $s$ is strictly greater than at least $\numplayer - k$ of the performances in $\signals_{-i}$ (since ties occur with probability zero by the continuous density assumption in \cref{a:fosd}).

The unconditional probability that player $i$ achieves rank $k$ or better, given their mean performance $\fitness_i$, is the expectation of $G_k$ over player $i$'s own performance distribution:
\begin{equation*}
    \cmf_{\fitness_i, \eqmfitnessesalt_{-i}}(k) = \int_{0}^{\infty} G_k(s \mid \eqmfitnessesalt_{-i}) \, d\distsignal_{\fitness_i}(s)
\end{equation*}

Now, consider $\fitness_i > \fitness'_i$. By \cref{a:fosd}, the distribution $\distsignal_{\fitness_i}$ first-order stochastically dominates $\distsignal_{\fitness'_i}$. Because increasing $s$ strictly expands the set of opponents' signal realizations $\signals_{-i}$ that $s$ can beat, the conditional probability function $G_k(s \mid \eqmfitnessesalt_{-i})$ is strictly increasing in $s$. Applying the FOSD property yields:
\begin{equation*}
    \int_{0}^{\infty} G_k(s \mid \eqmfitnessesalt_{-i}) \, d\distsignal_{\fitness_i}(s) \geq \int_{0}^{\infty} G_k(s \mid \eqmfitnessesalt_{-i}) \, d\distsignal_{\fitness'_i}(s),
\end{equation*}
which is exactly $\cmf_{\fitness_i, \eqmfitnessesalt_{-i}}(k) \geq \cmf_{\fitness'_i, \eqmfitnessesalt_{-i}}(k)$. 
\end{proof}

Next, we characterize the properties of the cost function $\fscostfn(\fitness,\type)$. For notational convenience, we drop the subscript $i$. 

\begin{lemma}\label{lemma:budget_property}
    Under Assumption \ref{a:sm}, $\fscostfn(\fitness,\type;\gamma)$ has the following properties:
    \begin{enumerate}
        \item $\fscostfn(0,\type;\gamma)=0$ for all $\type\in[\underline{\type},\bar{\type}]$.
        \item $\fscostfn(\fitness,\type;\gamma)$ is continuous in $\fitness$, $\type$, and $\gamma$.
        \item $\fscostfn(\fitness,\type;\gamma)$ is non-decreasing in $\fitness$ and non-increasing in both $\type$ and $\gamma$.
        \item $\fscostfn(\fitness,\type;\gamma)$ is convex in $\fitness$ for all $\type\in[\underline{\type},\bar{\type}]$.
        \item $\fscostfn(\fitness,\type;\gamma)$ is continuously differentiable with respect to $\fitness_i$ for all $\type_i\in[\underline{\type},\bar{\type}]$.
        \item $-\fscostfn(\fitness,\type;\gamma)$ has increasing differences in $(\fitness,\type)$. That is, for any $\lfitness<\ufitness$ and any $\type'<\type$, 
        $$
            \fscostfn(\ufitness,\type;\gamma)-\fscostfn(\lfitness,\type;\gamma)\leq\fscostfn(\ufitness,\type';\gamma)-\fscostfn(\lfitness,\type';\gamma) 
        $$
        \item $-\fscostfn(\fitness,\type;\gamma)$ has increasing differences with respect to $(\fitness,\gamma)$. Specifically, for any $\lfitness<\ufitness$ and any $\gamma'<\gamma$, 
        $$
            \fscostfn(\ufitness,\type;\gamma)-\fscostfn(\lfitness,\type;\gamma)\leq\fscostfn(\ufitness,\type;\gamma')-\fscostfn(\lfitness,\type;\gamma').
        $$
    \end{enumerate}
\end{lemma}

\begin{proof}[Proof of Lemma \ref{lemma:budget_property}]
We make the dependence on the AI-capability parameter $\gamma$ explicit in the cost and budget functions, but suppress it from intermediate expressions whenever it is held fixed.

Continuity follows from the Maximum Theorem and the continuity of $\mdeffortfn$, $\pteffortfn$, and the effort cost function $\etcostfn$. 
Monotonicity in $\fitness$ and $\type$ stems directly from the monotonicity of mean performance in both efforts and type. Monotonicity in $\gamma$ follows because a higher $\gamma$ raises the mean performance produced by any given allocation, strictly decreasing the minimal cost of attaining a fixed $\fitness$. Convexity in $\fitness$ holds because the upper contour sets of the concave production functions are convex.

Since $\etcostfn$ depends only on total effort and is monotone, $\fscostfn(\fitness,\type;\gamma)=\etcostfn(\budget(\fitness,\type;\gamma))$ where:
\begin{equation}\label{cost_minimization_equivalence}
    \budget(\fitness,\type;\gamma)=\min_{\mdeffort\geq 0,\pteffort\geq 0}\mdeffort+\pteffort \quad \text{ s.t. } \quad \mdeffortfn(\mdeffort,\type)+\pteffortfn(\pteffort;\gamma)=\fitness
\end{equation}

\paragraph{Continuous Differentiability.}
Because $\etcostfn(\cdot)$ is twice-differentiable, the Chain Rule implies $\fscostfn$ is continuously differentiable in $\fitness$ if and only if $\budget(\fitness, \type; \gamma)$ is continuously differentiable in $\fitness$. By the strict concavity of the production functions, the optimal effort allocation is unique for any $\fitness > 0$. We can therefore invoke the Envelope Theorem for constrained optimization. The value function $\budget(\fitness, \type; \gamma)$ is continuously differentiable with respect to $\fitness$, and its derivative is exactly the unique Lagrange multiplier $\lambda^*(\fitness)$ associated with the performance constraint:
$$
    \frac{\partial \budget(\fitness, \type; \gamma)}{\partial \fitness} = \lambda^*(\fitness)
$$
Because the constraint functions are continuously differentiable and the unique optimal solution varies continuously with $\fitness$, the shadow price $\lambda^*(\fitness)$ is continuous. Thus, $\budget(\fitness, \type; \gamma)$ and the composite function $\fscostfn(\fitness, \type; \gamma)$ are continuously differentiable in $\fitness$.

\paragraph{Increasing Differences in $(\fitness,\type)$.}
To establish that $-\budget(\fitness,\type)$ has increasing differences globally, we show that the marginal cost of performance, $\lambda^*(s, \type)$, is non-increasing in $\type$. At any optimal allocation for target $s$, the shadow price satisfies:
\begin{align}\label{eq:lambda_envelope}
    \lambda^*(s, \type) = \min \left\{ \frac{1}{\frac{\partial \mdeffortfn(\mdeffort^*, \type)}{\partial \mdeffort}}, \frac{1}{\pteffortfn'(\pteffort^*; \gamma)} \right\}
\end{align}
For any $\utype > \ltype$, the strategic complementarity of creative effort (Assumption \ref{a:sm}) guarantees that type $\utype$ has a strictly higher marginal product of creative effort. To achieve target $s$, $\utype$ must weakly shift effort toward the creative margin: $\mdeffort^*(\utype) \geq \mdeffort^*(\ltype)$ and $\pteffort^*(\utype) \leq \pteffort^*(\ltype)$. If both types exert no creative effort, $\lambda^*(s, \utype) = \lambda^*(s, \ltype)$. If $\utype$ exerts positive creative effort, the strict concavity of the production functions forces $\lambda^*(s, \utype) < \lambda^*(s, \ltype)$. By the Fundamental Theorem of Calculus, integrating this shadow price over $[\lfitness, \ufitness]$ guarantees that $\budget(\ufitness,\utype) - \budget(\lfitness,\utype) \leq \budget(\ufitness,\ltype) - \budget(\lfitness,\ltype)$. Since $\etcostfn(\cdot)$ is an increasing, convex transformation, this property transfers directly to $\fscostfn$.

\paragraph{Increasing Differences in $(\fitness,\gamma)$.}
Similarly, we show that $\lambda^*(s, \type; \gamma)$ is non-increasing in $\gamma$. For any $\ugamma > \lgamma$, the optimal creative effort must be non-increasing in $\gamma$ ($\mdeffort^*(\ugamma) \leq \mdeffort^*(\lgamma)$), otherwise the diverging marginal productivities would violate the optimality conditions. Through the strict concavity and supermodularity of $\pteffortfn$, evaluating the shadow price across all possible corner and interior solutions confirms that $\lambda^*(s, \type; \ugamma) \leq \lambda^*(s, \type; \lgamma)$ universally. Integrating over the performance interval transfers this property to the total budget and ultimately to the cost function $\fscostfn$.
\end{proof}

\begin{proof}[Proof of Proposition \ref{prop:exist}]
To evoke the sufficient conditions of \citet{reny2011existence}, we rely on our findings that the payoff function is bounded and continuous in $\fitness_i$, and exhibits increasing differences in $(\fitness,\type)$. Theorem 4.5 of \citet{reny2011existence} therefore guarantees the existence of a pure strategy Nash equilibrium that is symmetric and monotone.

To show that strict decreasing differences in the cost function eliminates pooling regions, suppose for contradiction that an interval of types $[\type_A, \type_B]$ pools on an interior performance level $\fitness^* > 0$. For $\fitness^*$ to be optimal for type $\type_A$, the first-order necessary condition must hold: $\frac{\partial \contestgain(\fitness^*, \eqmfitnessesalt^*)}{\partial \fitness} + 1 = \frac{\partial \fscostfn(\fitness^*, \type_A; \gamma)}{\partial \fitness}$. Under strict decreasing differences, type $\type_B$ faces a strictly lower marginal cost at $\fitness^*$. Therefore, type $\type_B$'s marginal utility at $\fitness^*$ is strictly positive, allowing them to achieve a higher payoff by unilaterally deviating to $\fitness^* + \varepsilon$. This contradicts the optimality of $\fitness^*$, proving the equilibrium must be strictly monotone.
\end{proof}

\begin{proof}[Proof of Proposition \ref{prop:ai_regimes}]
Because the prizes are bounded and the effort costs are convex, the equilibrium mean performance is bounded by some finite maximum $\fitness_{max}$. Let $\pteffort_{max}(\gamma) = \pteffortfn^{-1}(\fitness_{max}; \gamma)$.

\textbf{Part 1: Strict Separation for $\gamma < \underline{\gamma}$.} 
By Assumption \ref{a:gamma_limits}, there exists a threshold $\underline{\gamma}$ such that for all $\gamma < \underline{\gamma}$, $\pteffortfn'(0; \gamma) < \frac{\partial \mdeffortfn(0, \underline{\type})}{\partial \mdeffort}$. Under this condition, the marginal return to initial creative effort dominates mechanistic effort for all types $\type \in \Theta$. Thus, every agent utilizes strictly positive creative effort ($\mdeffort^* > 0$) for any $s > 0$. As shown in Lemma \ref{lemma:budget_property}, this forces the shadow price of performance to be strictly decreasing in type, granting $-\fscostfn$ strict increasing differences in $(\fitness, \type)$ and ensuring strict monotonicity.

\textbf{Part 2: Complete Pooling for $\gamma > \bar{\gamma}$.} 
As $\gamma \to \infty$, the marginal productivity of mechanistic effort approaches infinity. Thus, there exists a threshold $\bar{\gamma}$ such that for all $\gamma > \bar{\gamma}$, $\pteffortfn'(\pteffort_{max}(\gamma); \gamma) > \frac{\partial \mdeffortfn(0, \bar{\type})}{\partial \mdeffort}$. Consequently, the optimal effort allocation for \emph{all} types falls entirely into the pure-mechanization regime ($\mdeffort^* = 0$). The shadow price of the performance constraint becomes $\lambda^*(s, \type; \gamma) = 1/\pteffortfn'(\pteffortfn^{-1}(s; \gamma); \gamma)$, which is entirely independent of $\type$. Because the integrated incremental cost is identical across all types, $-\fscostfn$ satisfies increasing differences only with equality. Sharing the exact same objective function, all types share the same best response, resulting in a completely pooled equilibrium $\eqmfitnessalt^*(\type) = \fitness^*$.
\end{proof}
\section{Baseline}\label{proof:eqm_benchmark}

\begin{proof}[Proof of \cref{prop:eqm_benchmark} and \cref{coro:thresholds_gamma}]
The baseline problem ($\mathcal{P}_b$) simplifies to choosing efforts to maximize net utility. Holding $\gamma$ fixed, the agent solves:
\begin{equation}
    \max_{\mdeffort \geq 0, \pteffort \geq 0} \quad \mdeffortfn(\mdeffort, \type) + \pteffortfn(\pteffort; \gamma) - \etcostfn(\mdeffort + \pteffort).
\end{equation}
Because the objective is strictly concave (the sum of strictly concave production functions minus a strictly convex cost function), the necessary and sufficient Karush-Kuhn-Tucker (KKT) conditions uniquely characterize the global optimum $(\mdeffortbenchmark, \pteffortbenchmark)$:
\begin{align}
    \frac{\partial \mdeffortfn(\mdeffortbenchmark, \type)}{\partial \mdeffort} - \etcostfn'(\mdeffortbenchmark + \pteffortbenchmark) + \eta_1 &= 0, \label{eq:kkt_a}\\
    \pteffortfn'(\pteffortbenchmark; \gamma) - \etcostfn'(\mdeffortbenchmark + \pteffortbenchmark) + \eta_2 &= 0, \label{eq:kkt_b}
\end{align}
with complementary slackness $\eta_1 \mdeffortbenchmark = 0$, $\eta_2 \pteffortbenchmark = 0$, and $\eta_1, \eta_2 \geq 0$. Assumption \ref{a:pt} guarantees an interior total effort: $\mdeffortbenchmark + \pteffortbenchmark > 0$.

\textbf{Threshold Characterization:}
Let $\typebl$ be the marginal type indifferent between pure mechanization and entering the interior allocation. For $\type \leq \typebl$, $\mdeffortbenchmark = 0$, which implies $\eta_2 = 0$. By \eqref{eq:kkt_b}, $\pteffortbenchmark > 0$ uniquely solves $\pteffortfn'(\pteffortbenchmark; \gamma) = \etcostfn'(\pteffortbenchmark)$. The threshold $\typebl$ is defined where the non-negativity constraint on $\mdeffort$ just binds ($\eta_1 = 0$):
\begin{equation}\label{eq:thresh_1}
    \frac{\partial \mdeffortfn(0, \typebl)}{\partial \mdeffort} = \pteffortfn'(\pteffortbenchmark; \gamma).
\end{equation}
Since $\mdeffortfn$ is strictly supermodular (Assumption \ref{a:sm}), the marginal product $\frac{\partial \mdeffortfn(0, \type)}{\partial \mdeffort}$ is strictly increasing in $\type$. Thus, for all $\type < \typebl$, we must have $\eta_1 > 0 \implies \mdeffortbenchmark = 0$.

Similarly, let $\typebu$ be the marginal type indifferent between pure creation and the interior allocation. For $\type \geq \typebu$, $\pteffortbenchmark = 0 \implies \eta_1 = 0$. Here, $\mdeffortbenchmark > 0$ solves $\frac{\partial \mdeffortfn(\mdeffortbenchmark, \type)}{\partial \mdeffort} = \etcostfn'(\mdeffortbenchmark)$. The threshold $\typebu$ is defined where the non-negativity constraint on mechanistic effort just binds ($\eta_2 = 0$):
\begin{equation}\label{eq:thresh_2}
    \pteffortfn'(0; \gamma) = \frac{\partial \mdeffortfn(\mdeffortbenchmark(\typebu), \typebu)}{\partial \mdeffort}.
\end{equation}
For all $\type > \typebu$, supermodularity ensures the marginal return to creation strictly dominates, yielding $\eta_2 > 0 \implies \pteffortbenchmark = 0$. Because $\pteffortfn$ is strictly concave, $\pteffortfn'(0; \gamma) > \pteffortfn'(\pteffortbenchmark; \gamma)$, which, combined with the supermodularity of $\mdeffortfn$, guarantees $\typebu > \typebl$.

\textbf{Comparative Statics in $\gamma$ (Proof of \cref{coro:thresholds_gamma}):}
We now examine how $\typebl$ and $\typebu$ shift with the AI capability $\gamma$. 

For $\typebl(\gamma)$: Differentiating the identity $\pteffortfn'(\pteffortbenchmark; \gamma) = \etcostfn'(\pteffortbenchmark)$ with respect to $\gamma$ via the Implicit Function Theorem yields:
\begin{equation}
    \frac{d\pteffortbenchmark}{d\gamma} = \frac{\frac{\partial^2\pteffortfn}{\partial\pteffort\partial\gamma}}{\etcostfn''(\pteffortbenchmark) - \pteffortfn''(\pteffortbenchmark; \gamma)} > 0.
\end{equation}
Because $\etcostfn$ is convex and $\pteffortbenchmark$ is strictly increasing in $\gamma$, the marginal cost $\etcostfn'(\pteffortbenchmark)$ strictly increases. To maintain the equality in \eqref{eq:thresh_1}, the marginal product of creation $\frac{\partial \mdeffortfn(0, \typebl)}{\partial \mdeffort}$ must increase. By supermodularity, this requires the threshold $\typebl(\gamma)$ to be strictly increasing.

For $\typebu(\gamma)$: In \eqref{eq:thresh_2}, the left-hand side $\pteffortfn'(0; \gamma)$ is strictly increasing in $\gamma$ (since $\frac{\partial^2\pteffortfn}{\partial\pteffort\partial\gamma} > 0$). The right-hand side, $\etcostfn'(\mdeffortbenchmark(\typebu))$, depends on $\gamma$ only indirectly through $\typebu$. Since $\mdeffortbenchmark(\type)$ is strictly increasing in $\type$ (by monotone comparative statics on $\mdeffortfn$), the right-hand side is strictly increasing in $\typebu$. To maintain equality as the left-hand side rises, the threshold $\typebu(\gamma)$ must strictly increase. 
\end{proof}

\section{Contest Rewards and Incentives}

\begin{proof}[Proof of \cref{monotone strategy in reward}]
For any given strategy profile of other players $\eqmfitnessesalt_{-i}$ and type $\type_i$, we write the payoff function as:
\begin{align*}    
    \payoff(\fitness_i,\rewards) = \sum_{k=1}^\numplayer \pmf_{\fitness_i,\eqmfitnessesalt_{-i}}(k)\cdot \reward_k + \fitness_i - \fscostfn(\fitness_i,\type_i)
\end{align*}
For any performances $\fitness_i > \fitness'_i$ and any prize vectors $\rewards \geqq \rewards'$, we must show that the payoff exhibits increasing differences: $[\payoff(\fitness_i,\rewards)-\payoff(\fitness'_i,\rewards)] - [\payoff(\fitness_i,\rewards')-\payoff(\fitness'_i,\rewards')] \geq 0$. 

Canceling the deterministic terms, this cross-difference reduces entirely to the expected contest gains:
\begin{align*}
    \Delta &= \sum_{k=1}^\numplayer \left[ \pmf_{\fitness_i,\eqmfitnessesalt_{-i}}(k) - \pmf_{\fitness'_i,\eqmfitnessesalt_{-i}}(k) \right] (\reward_k - \reward'_k)
\end{align*}
Applying summation by parts to index $k$, we transform this into cumulative probabilities:
\begin{align*}
    \Delta &= \sum_{k=1}^{\numplayer-1} \underbrace{\left[ \cmf_{\fitness_i,\eqmfitnessesalt_{-i}}(k) - \cmf_{\fitness'_i,\eqmfitnessesalt_{-i}}(k) \right]}_{\text{Term A}} \cdot \underbrace{\left[ (\reward_k - \reward'_k) - (\reward_{k+1} - \reward'_{k+1}) \right]}_{\text{Term B}} 
\end{align*}
We evaluate the sign of each term:

\textbf{Term A:} By \cref{lemma:winprob_fosd}, higher performance yields a first-order stochastic dominance shift toward better (lower) ranks. Thus, the probability of achieving rank $k$ or better is weakly higher under $\fitness_i$: $\cmf_{\fitness_i}(k) - \cmf_{\fitness'_i}(k) \geq 0$.

    \textbf{Term B:} By the definition of the partial order on prize vectors (\cref{def:compare vectors}), $\rewards \geqq \rewards'$ implies that $\reward_k - \reward_{k+1} \geq \reward'_k - \reward'_{k+1}$. Rearranging this gives $(\reward_k - \reward'_k) - (\reward_{k+1} - \reward'_{k+1}) \geq 0$.

Because both terms are non-negative for all $k$, the entire summation $\Delta \geq 0$. 

Therefore, the payoff function $\payoff(\fitness_i,\rewards)$ exhibits increasing differences in $(\fitness_i,\rewards)$. By Theorem 4.5 of \citet{reny2011existence}, this guarantees the existence of a symmetric pure strategy Nash equilibrium whose mean performance $\eqmfitnessalt^*(\type, \rewards)$ is non-decreasing in the prize vector skewness $\rewards$.
\end{proof}

\begin{proof}[Proof of \cref{coro:fitness_gamma}]
The auxiliary-game payoff for type $\type_i$ choosing mean performance $\fitness_i$ is $\contestgain(\fitness_i,\eqmfitnessesalt_{-i}) + \fitness_i - \fscostfn(\fitness_i,\type_i;\gamma)$. By Property 7 of \cref{lemma:budget_property}, $-\fscostfn(\fitness_i,\type_i;\gamma)$ has increasing differences in $(\fitness_i,\gamma)$, and the remaining terms do not depend on $\gamma$; hence, the total payoff exhibits increasing differences in $(\fitness_i,\gamma)$. This mirrors the structure exploited in the proof of \cref{monotone strategy in reward}, with $\gamma$ taking the role of $\rewards$. Applying Theorem 4.5 of \citet{reny2011existence}, the symmetric monotone equilibrium mean performance $\fitnesscontest(\type)$ is non-decreasing in $\gamma$ for every $\type\in\Theta$.
\end{proof}

\section{Performance Manipulation}\label{apx:model hacking}

\begin{proof}[Proof of \cref{prop:compare}]

First, we establish that for any $\type' < \type$, if $\mdeffortcontest(\type) = 0$, then $\mdeffortcontest(\type') = 0$. 

Suppose for contradiction that there exists $\type' < \type$ such that $\mdeffortcontest(\type) = 0$ but $\mdeffortcontest(\type') > 0$. 
Because type $\type$ is at a corner solution where creative effort is zero, the marginal product of mechanization must weakly exceed that of creation evaluated at zero:
\begin{equation}\label{eq:corner_type}
    \pteffortfn'(\pteffortcontest(\type)) \geq \frac{\partial\mdeffortfn(0,\type)}{\partial\mdeffort}.
\end{equation}
By the monotonicity of the equilibrium, $\fitnesscontest(\type') \leq \fitnesscontest(\type)$. 
If type $\type'$ were to exert weakly greater mechanistic effort than type $\type$ ($\pteffortcontest(\type') \geq \pteffortcontest(\type)$), then because they also exert strictly positive creative effort ($\mdeffortcontest(\type') > 0 = \mdeffortcontest(\type)$), their total mean performance would be strictly higher: $\fitnesscontest(\type') > \fitnesscontest(\type)$. This violates equilibrium monotonicity. Thus, it must be that $\pteffortcontest(\type') < \pteffortcontest(\type)$.

By the strict concavity of $\pteffortfn$, this implies $\pteffortfn'(\pteffortcontest(\type')) > \pteffortfn'(\pteffortcontest(\type))$. 
Furthermore, by the strict supermodularity of $\mdeffortfn$ (Assumption \ref{a:sm}), the marginal productivity of creative effort is strictly increasing in type, and by strict concavity of $\mdeffortfn$, it is strictly decreasing in effort. Therefore:
\begin{equation*}
    \frac{\partial\mdeffortfn(0,\type)}{\partial\mdeffort} > \frac{\partial\mdeffortfn(0,\type')}{\partial\mdeffort} > \frac{\partial\mdeffortfn(\mdeffortcontest(\type'),\type')}{\partial\mdeffort}.
\end{equation*}
Chaining these inequalities with \eqref{eq:corner_type} yields:
\begin{equation*}
    \pteffortfn'(\pteffortcontest(\type')) > \pteffortfn'(\pteffortcontest(\type)) \geq \frac{\partial\mdeffortfn(0,\type)}{\partial\mdeffort} > \frac{\partial\mdeffortfn(\mdeffortcontest(\type'),\type')}{\partial\mdeffort}.
\end{equation*}
This strict inequality ($\pteffortfn' > \partial\mdeffortfn / \partial\mdeffort$) violates the interior optimality condition for type $\type'$; an optimizing agent would strictly prefer to shift effort away from creation and toward mechanization. This is a contradiction. Thus, $\mdeffortcontest(\type') = 0$.

We can therefore unambiguously define the threshold:
\begin{equation*}
    \typecl = \sup \left\{ \type : \pteffortfn'(\pteffortcontest(\type)) \geq \frac{\partial\mdeffortfn(0,\type)}{\partial\mdeffort} \right\}.
\end{equation*}
For all $\type < \typecl$, creative effort is zero. This directly establishes the first bullet point. 

Next, we establish that $\typecl \leq \typebl$.
Recall that in the baseline setting, any type $\type \leq \typebl$ only mechanizes, achieving $\fitnessbenchmark(\type) = \pteffortfn(\pteffortbenchmark(\type))$. 
By \cref{prop:benchmark and contest eqmfitness}, equilibrium mean performance is strictly higher in the contest: $\fitnesscontest(\type) \geq \fitnessbenchmark(\type)$. Therefore, for types exclusively mechanizing, $\pteffortcontest(\type) \geq \pteffortbenchmark(\type)$.
At the exact baseline threshold $\typebl$, the marginal product of mechanization exactly equals that of creation evaluated at zero effort:
\begin{equation*}
    \pteffortfn'(\pteffortbenchmark(\typebl)) = \frac{\partial\mdeffortfn(0,\typebl)}{\partial\mdeffort}.
\end{equation*}
Because $\pteffortcontest(\typebl) \geq \pteffortbenchmark(\typebl)$ and $\pteffortfn$ is strictly concave, we have:
\begin{equation*}
    \pteffortfn'(\pteffortcontest(\typebl)) \leq \pteffortfn'(\pteffortbenchmark(\typebl)) = \frac{\partial\mdeffortfn(0,\typebl)}{\partial\mdeffort}.
\end{equation*}
Comparing this inequality against the definition of $\typecl$, it immediately follows that $\typebl \geq \typecl$. This proves the second bullet point: types in $(\typecl, \typebl)$ mechanize in both environments, but only exert creative effort in the contest.

Finally, we prove the third bullet point. Consider any type $\type \in (\typebl, \typebu)$. In the baseline, this type is at an interior solution, meaning their optimal allocation satisfies:
\begin{equation*}
    \pteffortfn'(\pteffortbenchmark(\type); \gamma) = \frac{\partial\mdeffortfn(\mdeffortbenchmark(\type),\type)}{\partial\mdeffort}.
\end{equation*}
By \cref{prop:benchmark and contest eqmfitness}, $\fitnesscontest(\type) > \fitnessbenchmark(\type)$. Because both $\mdeffortfn$ and $\pteffortfn$ are strictly concave, achieving a strictly higher mean performance at minimum cost requires moving outward along the optimal expansion path where \emph{both} efforts strictly increase. If only one effort were to increase, its marginal product would strictly drop while the other's remained constant, violating the cost-minimization condition equating the marginal products. Thus, it must be that $\mdeffortcontest(\type) > \mdeffortbenchmark(\type) > 0$ and $\pteffortcontest(\type) > \pteffortbenchmark(\type) > 0$.

\end{proof}

\subsection{Examples of Performance Manipulation}\label{append:examples}

In this section, we present four concrete functional forms to illustrate how different production technologies dictate which agents engage in performance manipulation. Recall that an agent engages in performance manipulation if contest incentives induce them to weakly decrease (or hold constant) their creative effort ($\mdeffortcontest \leq \mdeffortbenchmark$) while strictly increasing their mechanistic effort ($\pteffortcontest > \pteffortbenchmark$).

\begin{itemize}
    \item In \textbf{Example \ref{ex:compare1}}, efforts are perfect substitutes. There are no interior ``middle'' types; low types manipulate, while high types do not.
    \item In \textbf{Example \ref{ex:compare2}}, the production function is Cobb-Douglas. All agents are middle types who scale both efforts proportionally to meet contest demands; thus, no types manipulate.
    \item In \textbf{Example \ref{ex:compare3}}, mechanization is strictly concave. Middle types meet higher contest demands by scaling only their creative effort, so they do not manipulate.
    \item In \textbf{Example \ref{ex:compare4}}, creation is exponential. Middle types meet higher contest demands by scaling only their mechanistic effort, meaning they do manipulate.
\end{itemize}

\begin{example}[Perfect Substitutes: No \emph{middle} types]\label{ex:compare1}
Consider $\numplayer\geq 2$ symmetric players. Let $\Theta= [0,3]$, assume that $\fitness=\type\mdeffort+\gamma\pteffort$, and let $\etcostfn(\budget)=\frac{1}{2}\budget^2$. Then $\typebl=\typebu=\typecl=\typecu=\gamma$.
\end{example} 

In this example, an agent engages solely in creation if $\type> \gamma$ and solely in mechanization if $\type<\gamma$. The minimal cost of attaining mean performance $\fitness$ for type $\type\geq \gamma$ is $\fscostfn(\fitness,\type)=\frac{1}{2}\left(\frac{\fitness}{\type}\right)^2$, and for type $\type<\gamma$ is $\fscostfn(\fitness,\type)=\frac{1}{2}\left(\frac{\fitness}{\gamma}\right)^2$. 

In the baseline scenario, an agent's optimized choice of $(\fitness,\mdeffort,\pteffort)$ is $(\type^2,\type,0)$ if $\type> \gamma$ and $(\gamma^2,0,\gamma)$ if $\type<\gamma$. In the contest, agents with $\type> \gamma$ continue to dedicate all effort to creation, and those with $\type<\gamma$ to mechanization.  
Because the two efforts are perfect substitutes, agents always opt entirely for the action where they exhibit higher marginal productivity. Applying \cref{prop:compare}, any type $\type<\gamma$ must increase their mechanistic effort to meet the higher contest performance, meaning they engage in performance manipulation. Types $\type>\gamma$ increase only creative effort and do not manipulate. 

\begin{example}[Cobb-Douglas: Only \emph{middle} types]\label{ex:compare2}
Consider $\numplayer\geq 2$ symmetric players. Let $\Theta = [0,\infty)$. Assume that $\fitness=\type\mdeffort^{\frac{1}{2}}+\gamma\pteffort^{\frac{1}{2}}$ and $\etcostfn(\budget) =\budget$. Then $\typebl=\typecl=0$, and $\typebu=\typecu=\infty$.
\end{example} 

In the baseline, an agent's optimal choice of $(\fitness,\mdeffort,\pteffort)$ is $\left(\frac{1}{2}(\type^2+\gamma^2), \frac{1}{4}\type^2, \frac{1}{4}\gamma^2\right)$. 
In the symmetric monotone equilibrium under the contest, any agent $i$ of any type $\type\in(0,\infty)$ engages in both actions. Moreover, the optimal ratio of the two efforts, $\frac{\mdeffortcontest(\type)}{\pteffortcontest(\type)} = \frac{\type^2}{\gamma^2}$, exactly coincides with the ratio in the baseline. This fixed-ratio property is driven by the Cobb-Douglas functional form of $\fitness$. 
Formally, to attain a strictly higher mean performance in the contest ($\fitnesscontest > \fitnessbenchmark$) while keeping the effort ratio fixed, the agent must strictly increase both efforts: $\mdeffortcontest(\type) > \mdeffortbenchmark(\type)$ and $\pteffortcontest(\type) > \pteffortbenchmark(\type)$. Because their creative effort strictly increases, no types engage in performance manipulation. 

\begin{example}[Strictly Concave Mechanization: No \emph{high} types]\label{ex:compare3}
Consider $\numplayer\geq 2$ symmetric players. Let $\Theta= [0,3]$, assume that $\fitness=\type\mdeffort+\gamma\sqrt{\pteffort}$, and let $\etcostfn(\budget)=\frac{1}{2}\budget^2$. 
\end{example} 

To achieve mean performance $\fitness$, a player of type $\type$ chooses:
\begin{align*}
    \mdeffort &= \begin{cases}
        0 & \text{if } \type\leq \frac{\gamma^2}{2\fitness} \\
        \frac{\fitness}{\type}-\frac{\gamma^2}{2\type^2} & \text{if } \type> \frac{\gamma^2}{2\fitness}
    \end{cases} \\
    \pteffort &= \begin{cases}
        \frac{\fitness^2}{\gamma^2} & \text{if } \type\leq \frac{\gamma^2}{2\fitness} \\
        \frac{\gamma^2}{4\type^2} & \text{if } \type> \frac{\gamma^2}{2\fitness}
    \end{cases}
\end{align*}
Hence, the derived cost function is:
\begin{equation*}
    \fscostfn(\fitness,\type)=\begin{cases}
            \frac{1}{2} \left(\frac{\fitness^2}{\gamma^2}\right)^2 & \text{ if } \type\leq \frac{\gamma^2}{2\fitness}\\[1ex]
            \frac{1}{2} \left(\frac{\fitness}{\type}-\frac{\gamma^2}{4\type^2}\right)^2  & \text{ if } \type> \frac{\gamma^2}{2\fitness}
    \end{cases}
\end{equation*}

In the baseline scenario, $\typebl=(\frac{\gamma}{2})^{\frac{2}{3}}$ and $\typebu=\infty$. If $\type\leq \typebl$, $\fitnessbenchmark(\type)=(\frac{\gamma^4}{2})^{\frac{1}{3}}$.
In the symmetric equilibrium under the contest, middle types do not engage in performance manipulation. Notice that in the interior regime ($\type> \frac{\gamma^2}{2\fitness}$), the optimal effort on mechanization ($\pteffort = \frac{\gamma^2}{4\type^2}$) depends only on the agent's type and $\gamma$, but is entirely independent of the target fitness $\fitness$. Therefore, as the contest induces a higher performance target, middle types achieve it by scaling only their creative effort, leaving mechanistic effort unchanged ($\pteffortcontest = \pteffortbenchmark$).

\begin{example}[Exponential Creation: \emph{Middle} types manipulate]\label{ex:compare4}
Consider $\numplayer\geq 2$ symmetric players. Let $\Theta=[0,9]$. Assume that $\fitness=\type(1-\exp{(-\mdeffort)})+\gamma\pteffort$, and $\etcostfn(\budget)=\frac{1}{4} \budget^2$.  
\end{example}

The optimal effort on creation is given by:
\begin{equation*}
    \mdeffort = \begin{cases}
        0 & \text{if } \type\leq \gamma \\
        \ln{\left(\frac{\type}{\gamma}\right)} & \text{if } \gamma<\type\leq \fitness+\gamma \\
        \ln{\left(\frac{\type}{\type-\fitness}\right)} & \text{if } \type> \fitness+\gamma
    \end{cases}
\end{equation*}
The corresponding effort on mechanization is:
\begin{equation*}
    \pteffort = \begin{cases}
        \frac{\fitness}{\gamma} & \text{if } \type\leq \gamma \\
        \frac{\fitness-\type+\gamma}{\gamma} & \text{if } \gamma< \type\leq\fitness+\gamma \\
        0 & \text{if } \type>\fitness+\gamma
    \end{cases}
\end{equation*}
This yields the cost function:
\begin{equation*}
    \fscostfn(\fitness,\type)=\begin{cases}
        \frac{1}{4} \left(\frac{\fitness}{\gamma}\right)^2 & \text{if } \type<\gamma \\[1ex]
        \frac{1}{4} \left(\frac{\fitness-\type+\gamma}{\gamma} + \ln{\left(\frac{\type}{\gamma}\right)}\right)^2 & \text{if } \gamma\leq \type<\gamma+\fitness \\[1ex]
        \frac{1}{4} \left(\ln{\left(\frac{\type}{\type-\fitness}\right)}\right)^2 & \text{if } \type>\fitness +\gamma
    \end{cases}
\end{equation*}
 
In the baseline scenario, $\fitnessbenchmark(\type)=2\gamma^2 $ if $\type <\gamma$; $\fitnessbenchmark(\type)=2\gamma^2-\gamma\ln{\left(\frac{\type}{\gamma}\right)}+\type-\gamma $ if $\gamma\leq\type <\gamma e^{2\gamma}$; and $\fitnessbenchmark$ is implicitly defined by $2(\type-\fitnessbenchmark(\type))=\ln{\left(\frac{\type}{\type-\fitnessbenchmark(\type)}\right)}$ if $\type \geq \gamma e^{2\gamma}$. 
Thus, $\typebl=\gamma$ (players below $\gamma$ only mechanize) and $\typebu=\gamma e^{2\gamma}$ (players above this strictly create), with middle types utilizing both actions.

In the symmetric equilibrium under the contest, middle types \emph{do} engage in performance manipulation. Notice that in the interior regime, the optimal choice of creative effort ($\mdeffort = \ln(\type/\gamma)$) is entirely independent of the target performance $\fitness$. Because mean performance is strictly higher in the contest than in the baseline, middle types achieve this higher target strictly by scaling their mechanistic effort ($\pteffort = \frac{\fitness-\type+\gamma}{\gamma}$), holding their creative effort constant ($\mdeffortcontest = \mdeffortbenchmark$).

\end{document}